\documentclass[journal]{IEEEtran}
\usepackage{amsmath,amsfonts}
\usepackage{array}
\usepackage[caption=false,font=normalsize,labelfont=sf,textfont=sf]{subfig}
\usepackage{textcomp}
\usepackage{extarrows}
\usepackage{stfloats}
\usepackage{url}
\usepackage{verbatim}
\usepackage{graphicx}
\usepackage{cite}
\usepackage{balance}
\usepackage{amssymb}
\usepackage{amsthm}
\usepackage{mathrsfs}
\usepackage{indentfirst}
\usepackage{multirow}
\usepackage{algorithm}
\usepackage{algpseudocode}
\usepackage{bm}
\usepackage{xcolor}
\usepackage{multicol}
\usepackage[font=scriptsize]{caption}
\usepackage{adjustbox}
\usepackage{tabularx}
\usepackage{diagbox}
\usepackage{makecell}
\usepackage{algpseudocode}
\theoremstyle{remark}

\newtheorem{corollary}{Corollary}
\newtheorem{theorem}{Theorem}
\newtheorem{lemma}{Lemma}
\begin{document}
\title{HARQ-IR Aided Short Packet Communications: BLER Analysis and Throughput Maximization}
\author{Fuchao He, Zheng Shi, Guanghua Yang, Xiaofan Li, Xinrong Ye, and Shaodan Ma
\thanks{Fuchao He, Zheng Shi, Guanghua Yang, and Xiaofan Li are with the School of Intelligent Systems Science and Engineering and GBA and B\&R International Joint Research Center for Smart Logistics, Jinan University, Zhuhai 519070, China (e-mails: hefuchao@stu2022.jnu.edu.cn, shizheng0124@gmail.com, ghyang@jnu.edu.cn, lixiaofan@jnu.edu.cn).}
\thanks{Xinrong~Ye is with the School of Physics and Electronic Information, Anhui Normal University, Wuhu 241002, China (e-mail: shuchong@mail.ahnu.edu.cn).}
\thanks{Shaodan Ma is with the State Key Laboratory of Internet of Things for Smart City and the Department of Electrical and Computer Engineering, University of Macau, Macao SAR, China (e-mail: shaodanma@um.edu.mo).}
}

\maketitle
\begin{abstract}

This paper introduces hybrid automatic repeat request with incremental redundancy (HARQ-IR) to boost the reliability of short packet communications. The finite blocklength information theory and correlated decoding events tremendously preclude the analysis of average block error rate (BLER). Fortunately, the recursive form of average BLER motivates us to calculate its value through the trapezoidal approximation and Gauss–Laguerre quadrature. Moreover, the asymptotic analysis is performed to derive a simple expression for the average BLER at high signal-to-noise ratio (SNR). Then, we study the maximization of long term average throughput (LTAT) via power allocation meanwhile ensuring the power and the BLER constraints. For tractability, the asymptotic BLER is employed to solve the problem through geometric programming (GP). However, the GP-based solution underestimates the LTAT at low SNR due to a large approximation error in this case. Alternatively, we also develop a deep reinforcement learning (DRL)-based framework to learn power allocation policy. In particular, the optimization problem is transformed into a constrained Markov decision process, which is solved by integrating deep deterministic policy gradient (DDPG) with subgradient method. The numerical results finally demonstrate that the DRL-based method outperforms the GP-based one at low SNR, albeit at the cost of increasing computational burden.

\end{abstract}

\section{Introduction}
\subsection{Background}

The timeliness of information transmission has recently become more and more important especially in the future highly interconnected society \cite{ref30}. To fulfill this goal, short packet communications have emerged as one of the key technologies in sixth generation (6G) wireless communication systems that enable real-time mobile communication services/applications, e.g., sensory interconnection, factory automation, autonomous driving, metaverse, digital twins, holographic communications, etc. \cite{ref1,ref2}. However, the classical Shannon information theory is established under the assumption of infinite blocklength, which is no longer valid in short packet communications. Thereupon, Yury Polyanskiy {\emph {et al.}} in \cite{ref29} developed the finite blocklength information theory, which rigorously substantiates the deterioration of reliability performance with the shortening of the packet length.




\subsection{Related Works}
The reliability performance of various short packet communication systems has been intensively reported in the literature. To name a few, in \cite{ref47}, the average block error rate (BLER) of short packet communications over independent multiple-input multiple-output (MIMO) fading channels was approximately obtained by using Laplace transform. In \cite{ref49}, the average BLER of short packet vehicular communications over correlated cascaded Nakagami-$m$ fading channels was derived by capitalizing on Mellin transform. A partial cooperative non-orthogonal multiple access (NOMA) protocol was proposed to assist short packet communications in \cite{ref50}, in which the asymptotic BLER analysis was performed to reveal engineering insights. Moreover, by considering the hardware impairments and channel estimation errors in \cite{ref51}, the BLER of the NOMA-assisted short packet communications was analyzed. Furthermore, by integrating the intelligent reflecting surface (IRS) with NOMA, the BLER of short packet communications was obtained by considering random and optimal phase shifts \cite{ref52}. In addition, by taking into account the possibility of failed elements at IRS, the asymptotic BLER was further derived in \cite{ref53}, with which the performance loss caused by phase errors and failure elements was thoroughly investigated. All the prior literature justify the contradictory relationship between reliability and latency, that is, the low latency of short packet communications is achieved at the cost of degraded reliability.

In order to enhance the reception reliability of short packet communications, hybrid automatic repeat request (HARQ) is a promising technology to flexibly balance the latency and the reliability of communications \cite{ref19}. In general, HARQ can be categorized into three basic types according to different coding/combining techniques, including Type-I HARQ, HARQ with chase combining (HARQ-CC), and HARQ with incremental redundancy (HARQ-IR). There are only a few existing works that studied HARQ-assisted short packet communications. Particularly, Type-I HARQ and HARQ-CC-aided short packet communications
were examined in \cite{ref5}, where BLER was derived by assuming Nakagami-$m$ fading channels and ignoring the correlation among decoding failures. Moreover, in \cite{ref10}, a dynamic NOMA-HARQ-CC scheme was devised to schedule the transmissions of the old and fresh status updates, where the optimal power sharing fraction is chosen by invoking Markov decision process (MDP). Dileepa Marasinghe {\emph {et al.}} in \cite{ref6} studied the block error performance of NOMA-HARQ-CC assisted short packet communications by only considering the decoding failure in the current HARQ round. 
In \cite{ref9}, the BLER of HARQ-CC and HARQ-IR aided short packet communications over constant fading channels were derived by using Markov chain. By assuming quasi-static fading channels, \cite{ref8} derived the average BLER of HARQ-IR-aided short packet communications by using linear approximation technique. In addition, the average BLER of HARQ-IR-aided short packet communications was obtained by using Gaussian approximation in \cite{ref36}, with which dynamic programming was invoked to minimize the average power subject to the latency constraint. The similar method was extended to investigate the throughput maximization in \cite{ref35}. Moreover, bearing in mind the Gaussian approximation, the authors in \cite{ref34} studied the effective capacity of HARQ, which enabled the minimization of average power while maintaining the minimum effective capacity. Unfortunately, the above literature derived the average BLER of HARQ-IR-aided short packet communications based on either the oversimplified assumption (e.g., quasi-static fading, independent decoding failures, etc) or approximation (e.g., Gaussian approximation), which might underestimate the actual system performance. Accordingly, this paper focuses on the accurate performance evaluation and effective design of HARQ-IR-aided short packet communications.

\subsection{Main Contributions}
This paper examines the average BLER and LTAT maximization of HARQ-IR-aided short packet communications. To recapitulate, the contributions of this paper can be briefly summarized as follows:
\begin{itemize}
    \item The finite blocklength (FBL) information theory and the correlation between decoding events across different HARQ rounds considerably challenge the analysis of the average BLER. Fortunately, by noticing the recursive integration form of the average BLER, the trapezoidal approximation method and Gauss–Laguerre quadrature method are proposed to numerically evaluate the average BLER. Besides, the dynamic programming is leveraged to further reduce the complexity of the Gauss–Laguerre quadrature method that originates from the avoidance of redundant calculations. Furthermore, the asymptotic analysis is conducted to offer a simple approximate expression for the average BLER in the high SNR regime.

    \item This paper investigates the power allocation scheme of HARQ-IR-aided short packet communications. The goal is maximizing the  LTAT while guaranteeing the power and the BLER constraints. With the asymptotic results, the optimization problem can be converted into a convex problem through the geometric programming (GP). Whereas, the GP-based solution underestimates the system performance in low SNR. This is due to the fact that there is a large approximation error between the exact and the asymptotic BLER.


    \item To address this issue, the tool of deep reinforcement learning (DRL) is utilized to learn power allocation policy from the environment. In doing so, the optimization problem is first transformed into a constrained Markov decision process (MDP) problem, which can be solved via the combination of DDPG framework and subgradient method. In addition, the truncation-based updating strategy is proposed to warrant the stable convergence of its training process. The numerical results finally show that the DRL-based method is superior to the GP-based one in sacrificed of high computational burden during offline training stage.


\end{itemize}

\subsection{Structure of This Paper}
The rest of the paper is outlined as follows. Section \ref{SYSTEM MODEL} introduces the system model of HARQ-IR aided short packet communications. The average BLER is derived in Section \ref{PERFORMANCE ANALYSIS}. Section \ref{Throughput Optimization} studies the maximization of LTAT. Section \ref{SIMULATION ANALYSIS} carries out the numerical analysis for verification. Finally, Section \ref{CONCULSION} concludes this paper.

\section{System Model}\label{SYSTEM MODEL}

Short packet transmission is one of the key enabling technologies to ensure low latency. However, according to the finite blocklength coding theory, the successful reception cannot be guaranteed even if the channel capacity is larger than the transmission rate, which leads to the degradtion of the reliability. To overcome this issue, HARQ-IR is adopted to increase the reliability of short packet communications. In the following, the proposed HARQ-IR aided short packet transmission scheme and the performance metric of reliability, i.e., block error rate (BLER), are introduced.



\subsection{HARQ-IR-Aided Short Packet Transmissions}
By following the principle of HARQ-IR, each message consisting of $\mathcal{K}$ information bits is first encoded into a mother codeword, which is then chopped into $M$ short sub-codewords each with a length of $L$. Clearly, the number of transmissions for each message is allowed up to $M$ and an outage event is declared if the maximum number of transmissions is exceeded. Mathematically, by assuming block fading channels, the received signal at the $m$-th HARQ round is expressed as
\begin{equation}
    \bold{y}_m = \sqrt{P_m}h_m\bold{x}_m + \bold{n}_m,
\end{equation}
where $P_m$, $\bold{x}_m$, ${\bf n}_m$, and $h_m$ represent the transmission power, the transmitted sub-codeword, the complex additive Gaussian noise, and the channel coefficient at the $m$-th HARQ round, respectively. More specifically, we assume that $h_m$ obeys Rayleigh fading and ${\bf n}_m$ follows normal distribution with mean zero and variance $\mathcal N_0$, i.e., $\bold{n}_m \sim \mathcal{CN}(0, \mathcal N_0)$. Once the receiver fails to decode the message in the $m$-th HARQ round, a negative acknowledgement message (NACK) is fed back to the sender to request a retransmission. Only if the receiver successfully decodes the message or the maximum number of transmissions is reached, the next new message is delivered in subsequent transmission. 
The transmission protocol of HARQ-IR aided short packet communications is illustrated in Fig. \ref{HARQ-IR}.





\begin{figure}[h]
\centerline{\includegraphics[width=1\linewidth]{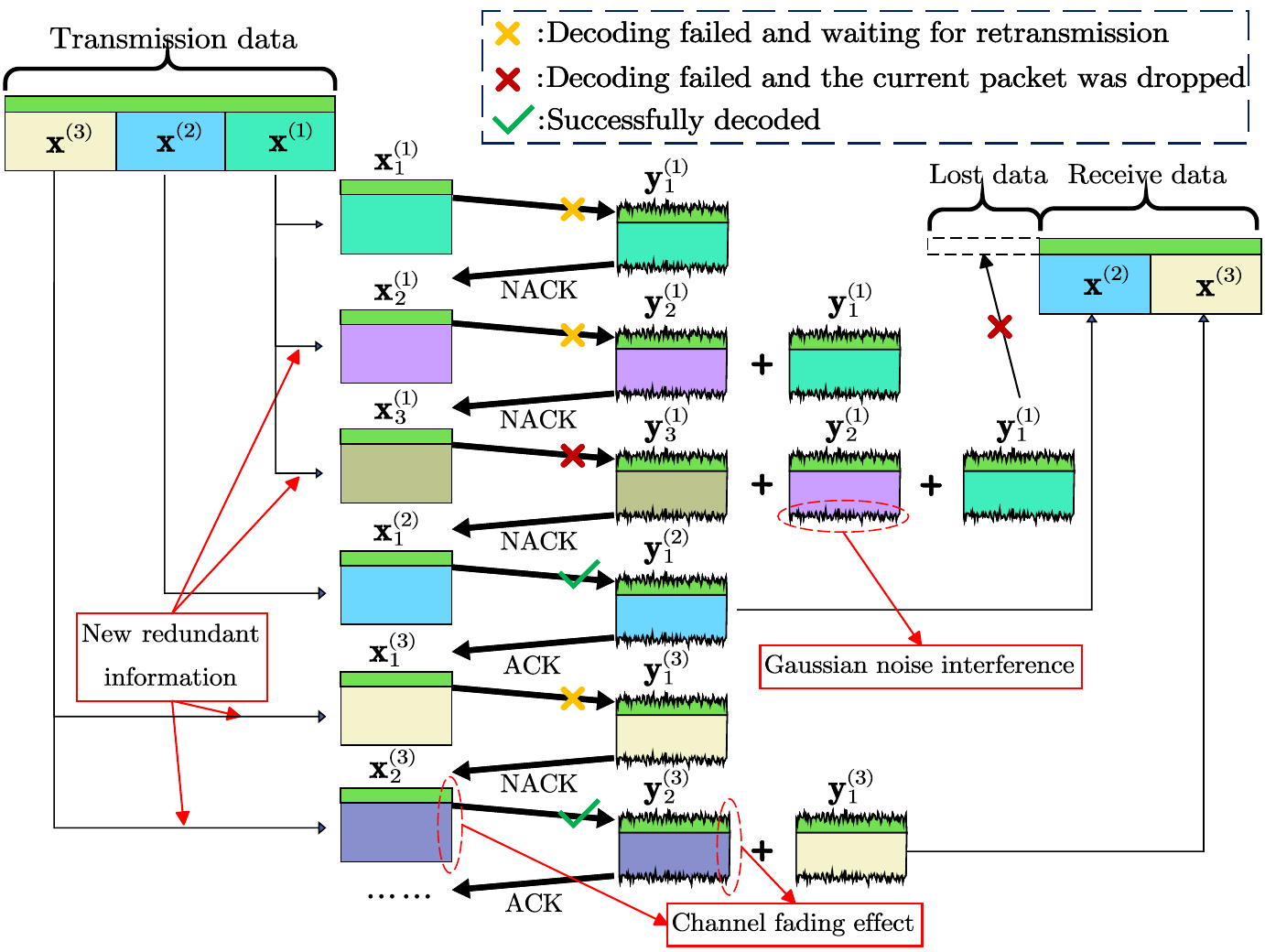}}
\caption{An example of HARQ-IR aided short packet communications with $M=3$.}
\label{HARQ-IR}
\end{figure}

\subsection{BLER}
A failure $\mathcal F_{M}$ of delivering a message takes place if and only if it cannot be reconstructed within $M$ HARQ rounds. According to the finite blocklength (FBL) information theory, conditioned on the channel gains $g_m = |h_m|^2$ for $m=1,\cdots,M$ and failures at the first $M-1$ HARQ rounds, the conditional BLER $\epsilon_M$ after $M$ HARQ rounds is given by \cite{ref27}
\begin{align}\label{epsilon-m}
        \epsilon_M &= \Pr\{ \mathcal F_{M} | g_1,\cdots, g_M, \mathcal F_{M-1}\} \notag\\
        &\approx Q\Bigg(\frac{\sum_{m=1}^M \log_2(1 + \bar \gamma_m g_m) - R}{\log_2 e\sqrt{\frac{1}{L}\sum_{m=1}^M \big(1 - (1 + \bar \gamma_m g_m)^{-2} \big)}}\Bigg),
\end{align}
where $\mathcal F_{m}$ denotes the decoding failure at the $m$-th HARQ round, $\bar \gamma_m = P_m/\mathcal N_0$ is the average transmit SNR at the $m$-th HARQ round, $R=\mathcal{K}/L$ refers to the initial transmission rate, 
and $Q(x) = \int_{x}^{+\infty}(\sqrt{2\pi}e^{{t^2}/{2}})^{-1}dt$. Besides, we stipulate that $\epsilon_0 = 1$. On the basis of \eqref{epsilon-m} along with Bayes' formula, the average BLER after $M$ HARQ rounds is obtained as
\begin{equation}\label{origin_real}
    \begin{aligned}
        \mathbb{\bar P}_M =& \Pr\{\mathcal F_{1},\cdots,\mathcal F_{M}\} \\
        =& \mathbb{E}_{g_1,\cdots,g_M}\Bigg[\prod_{m=1}^{M}\epsilon_{m} \Bigg], \\
    \end{aligned}
\end{equation}
where $\epsilon_{m}$ can be obtained by replacing $M$ in \eqref{epsilon-m} with $m$, $\mathbb{E}[\cdot]$ stands for the expectation operation. By noticing Rayleigh fading channels, the channel gain $g_m$ follows exponential distribution with probability density function (PDF) given by
\begin{equation} \label{channel gain}
    \begin{aligned}
        f_{g_m}(x) = {\lambda}^{-1} e^{-{\lambda}^{-1}x} ,\, x > 0,
    \end{aligned}
\end{equation}
where $\lambda$ denotes the average power of channel gain.

\section{Analysis of Average BLER}\label{PERFORMANCE ANALYSIS}
Due to the fractional form in $Q$ function and the correlation among $\epsilon_{m}$ caused by the same channel gain $g_m$, it is fairly intractable to derive a closed-form expression for \eqref{origin_real}. Hence,  this section is devoted to numerically evaluated the average BLER $\mathbb{\bar P}_M$ of HARQ-IR aided short packet communications.

In order to facilitate the computation of the average BLER, we apply an approximation $\sqrt{\sum_{m=1}^{M}\big(1 - {(1 + \bar \gamma_m g_m)^{-2}} \big)} \approx \sqrt{M}$ to (\ref{epsilon-m}), which yields
\begin{equation}\label{eqn:eps_Mapp}
    \epsilon_M \approx Q\Bigg(\frac{\sum_{m=1}^M \log_2(1 + \bar \gamma_m g_m) - R}{\sqrt{{M}}V}\Bigg),
\end{equation}
where $V = \sqrt{{1}/{L}}\log_2 e$ and this approximation is valid in the high SNR regime \cite{ref32}. Similarly to \eqref{eqn:eps_Mapp}, we can also obtain the approximate expressions of $\epsilon_1,\cdots,\epsilon_{M-1}$. Thus, the average BLER of the $M$-th transmission is approximated as
\begin{equation}
    \begin{aligned}\label{E_epsilon'_m}
        &\mathbb{\bar P}_M \approx \mathbb{E}_{g_1,\cdots,g_M}\Big[\mathcal{Q}_1(C_1)\cdots \mathcal{Q}_M(C_1+\cdots C_M)\Big]\\
       & =\int\limits_0^{+\infty}\cdots\int\limits_0^{+\infty}\prod_{m=1}^M \mathcal{Q}_m\left(\sum\limits_{i=1}^{m}x_i\right) f_{C_m}(x_m) dx_1\cdots dx_M,
    \end{aligned}
\end{equation}
where $\mathcal{Q}_{m}(y) = Q({y - R}/({\sqrt{m}V}))$, $C_m = \log_2(1 + \bar \gamma_m g_m)$, and $f_{C_m}(x)$ denotes the PDF of $C_m$. Particularly, according to \eqref{channel gain}, $f_{C_m}(x)$ is obtained as
\begin{equation}
    \begin{aligned}\label{eqn:f_cm_pdf}
        f_{C_m}(x) 
        &=\frac{\lambda^{-1}\ln2}{\bar \gamma_m}2^{x}e^{-(2^{x} - 1){\bar \gamma_m^{-1}}\lambda^{-1}}, & x > 0.\\
    \end{aligned}
\end{equation}

By making change of variables $S_m = \sum_{i=1}^m C_i$ for $m\in[1,M]$ and leveraging Jacobian transform, \eqref{E_epsilon'_m} can be rewritten as
\begin{align}\label{E_epsilon'_m1}
        \mathbb{\bar P}_M \approx&\int_{S_0}^{+\infty}\cdots\int_{S_{M-1}}^{+\infty}\prod_{m=1}^M\mathcal{Q}_m(S_m)f_{C_m}(S_m-S_{m-1})\notag\\
        &\times \det(\bold{J}_{\bold{C}\rightarrow\bold{S}})dS_1\cdots dS_M,
    \end{align}
wherein $S_0 = 0$, ${\bf C}=(C_1,\cdots,C_M)$, ${\bf S}=(S_1,\cdots,S_M)$, and the Jacobian transform matrix is given by
\begin{equation}
    \begin{aligned}
    \bf{J}_{\bold{C}\rightarrow\bold{S}}&= \left(\frac{\partial C_i}{\partial S_j} \right)=
        \left( {\begin{array}{*{20}{c}}
1&0&0& \cdots &0\\
1&1&0& \cdots &0\\
0&1&1& \cdots &0\\
 \vdots & \ddots & \ddots & \ddots & \vdots \\
0&0&0&1&1
\end{array}} \right).
    \end{aligned}
\end{equation}
Thus we have $\det(\bold{J}_{\bold{C}\rightarrow\bold{S}}) = 1$. However, it is intractable to derive the multifold integral \eqref{E_epsilon'_m1} in closed-form. Fortunately, it is easily seen by inspection that \eqref{E_epsilon'_m1} can be rewritten in a recursive form as
\begin{align}\label{bar_BLER_origin}
        \mathbb{\bar P}_M =&\int_{S_0}^{+\infty}\cdots\int_{S_{M-1}}^{+\infty}\prod_{m=1}^M\mathcal{Q}_m(S_m)f_{C_m}(S_m-S_{m-1})\notag\\
        &\times \phi_M(S_M)dS_1\cdots dS_M\notag\\
        =&\int_{S_0}^{+\infty}\cdots\int_{S_{M-2}}^{+\infty}\prod_{m=1}^{M-1}\mathcal{Q}_m(S_m)f_{C_m}(S_m-S_{m-1})\notag\\
        &\times \phi_{M-1}(S_{M-1})dS_1\cdots dS_{M-1}\notag\\
        =& \cdots =\int_{S_0}^{+\infty} \mathcal{Q}_1(S_1)f_{C_1}(S_1-S_0) \phi_{1}(S_1)dS_1 \notag\\
        =&  \phi_0(S_0),
\end{align}
where ${\phi _m}({S_m})$ is defined in \eqref{phi_sm} as shown at the top of this page.
\begin{figure*}
\begin{equation}\label{phi_sm}
{\phi _m}({S_m}) = \left\{ {\begin{array}{*{20}{l}}
{\int_{{S_m}}^{ + \infty } {{{\cal Q}_{m + 1}}} ({S_{m + 1}})f_{C_{m+1}}({S_{m + 1}} - {S_m}){\phi _{m + 1}}({S_{m + 1}})d{S_{m + 1}},}&{0 \le m \le M - 1}\\
1,&{m = M}
\end{array}} \right. .
\end{equation}
\hrulefill
\end{figure*}
The recursive relationship among ${\phi _m}({S_m})$ motivates us to convert the multifold integral \eqref{E_epsilon'_m1} into $M$ successive single-fold integrals. There are many sophisticated numerical integration methods that can be utilized to iteratively compute the single-fold integral in \eqref{phi_sm}. To do so, several numerical integration methods are developed to evaluate $\mathbb{\bar P}_M$ in the following, including trapezoidal approximation, Gauss–Laguerre quadrature, and dynamic programming.
\subsection{Trapezoidal Approximation Method}\label{sec:trap}
The essence of trapezoidal rule for integration is dividing the integration domain into $K$ intervals, i.e, $[a_k, a_{k+1}), k = 1,\cdots,K$, where each interval has the same length $\mathcal{H}$. To realize the calculation of \eqref{phi_sm} with trapezoidal approximation, the integral defined in \eqref{phi_sm} is truncated from an infinite interval into a finite one, i.e.,
\begin{align}\label{G_L}
         \phi_{m}(S_{m})  &= \int_{0}^{\infty}\mathcal{Q}_{m+1}(S_{m}+x)f_{C_{m+1}}(x)\phi_{m+1}(S_{m}+x)dx\notag\\
        &\approx \int_{0}^{U}\psi_{m}(x)dx\notag\\
        &\approx \frac{\mathcal{H}}{2}\sum_{k=1}^{K}\Big(\psi_{m}((k-1)\mathcal H) + \psi_{m}(k\mathcal H)\Big),
\end{align}
where the first equality holds by change of variable $x = S_{m+1} - S_m $, $U$ denotes the upper bound of the truncated integration interval, and  $\psi_{m}(x)$ is defined as
\begin{equation}\label{eqn:psi_m_x}
    \begin{aligned}
        \psi_{m}(x) =& \mathcal{Q}_{m+1}(S_{m} + x)f_{C_{m+1}}(x)\phi_{m+1}(S_{m} + x).
    \end{aligned}
\end{equation}
It is clear from \eqref{G_L} that $U = K \mathcal H$ follows. The discussion of the truncation error in \eqref{G_L} is deferred to the next subsection.

However, in order to calculate the average BLER $\mathbb{\bar P}_M$, the trapezoidal approximation method needs to integrate out $S_M,\cdots,$ and $S_1$ layer by layer according to \eqref{G_L}. To enable such layer-by-layer numerical evaluation, it is necessary to approximate $\phi_{m}(S_{m})$ with a discrete function. Henceforth, it is intuitive to discretize the function of $\phi_{m}(S_{m})$ with the same interval $\mathcal H$ to avoid possible necessity of curve fitting.
Furthermore, the trapezoidal approximation algorithm should be well designed to avoid repeated calculations of the intermediate values $\phi_{m}(x)$. For example, it involves the same evaluations of the values $\psi_{m}(2\mathcal H),\cdots,\psi_{m}((K+1)\mathcal H)$ when we compute $\phi_{m}(\mathcal H)$ and $\phi_{m}(2 \mathcal H)$.

To calculate $\phi_0(S_0)$, it suffices to obtain $(K+1)$ values of $\psi_{1}(x)$ (i.e., $\psi_{1}(0),\psi_{1}(\mathcal H)\cdots,\psi_{1}(K\mathcal H)$) according to the trapezoidal approximation in \eqref{G_L}. For each value of $\psi_{1}(x)$, we first have to get the values of $\phi_1(x)$ at $x=0,\mathcal H,\cdots, K\mathcal H$. By analogy, to obtain these values, we further need to calculate $2(K+1)$ values of $\psi_{2}(x)$, i.e., $\psi_{2}(0),\psi_{2}(\mathcal H)\cdots,\psi_{2}(K\mathcal H),\cdots,\psi_{2}(2K\mathcal H)$. Thereon, $\tilde{K} = M(K+1)$ initial values of $\psi_{M}(x)$ should be obtained in the first layer of integration, i.e., $\psi_{M}(0),\psi_{M}(\mathcal H),\cdots, \psi_{M}(\tilde K \mathcal H)$, wherein $\phi_{M}(0)=\phi_{M}(\mathcal H)=\cdots=\phi_{M}(\tilde K \mathcal H)=1$ are needed for the calculation. The computational procedure of the proposed trapezoidal approximation method is summarized in Fig. \ref{trapezoidal_approximation_algorithm}, as shown at the top of the next page.
\begin{figure*}[h]
\centerline{\includegraphics[width=1\linewidth]{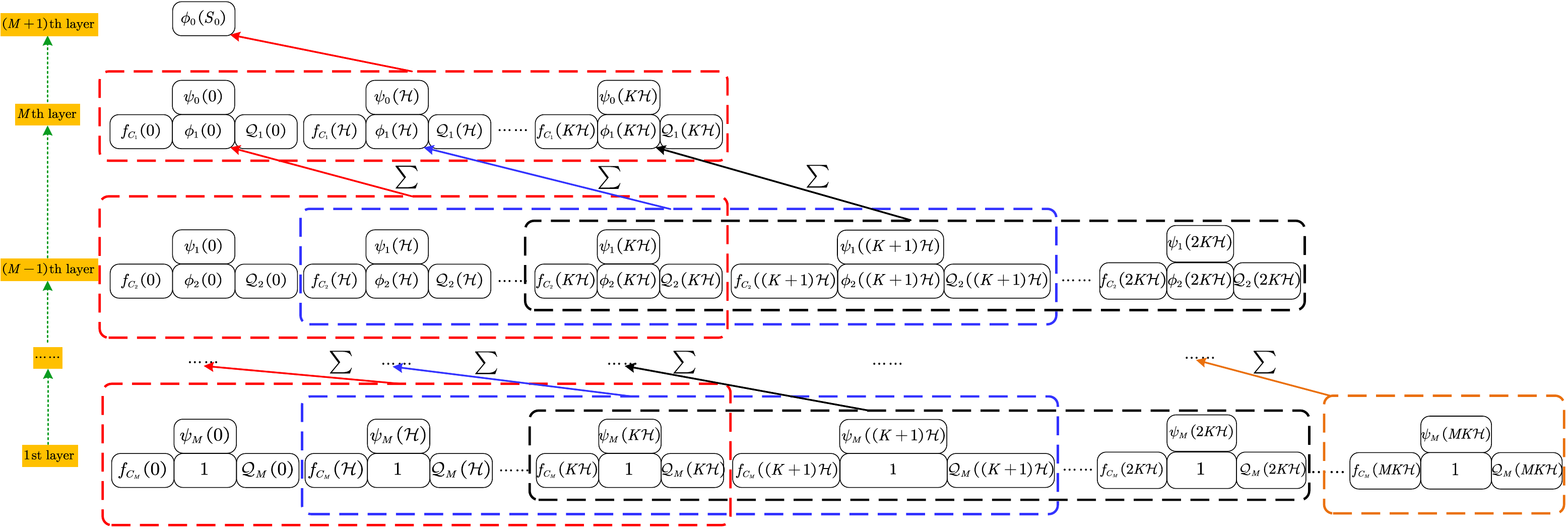}}
\caption{Computational procedure of trapezoidal approximation.}
\label{trapezoidal_approximation_algorithm}
\end{figure*}
\subsubsection{Error Analysis}
Clearly from \eqref{G_L}, the approximation error consists of two parts, i.e., truncation error and discretization error. More specifically, the truncation operation in \eqref{G_L} would inevitably yield an approximation error.
To choose an appropriate truncation value of $U$, the truncation error $\mathcal E_U$ of \eqref{G_L} should be analyzed. According to \eqref{G_L}, since $\mathcal{Q}_{m}(y) \le 1$ and $\phi_{m}(x) \le 1$, $\mathcal E_U$ is upper bounded as
\begin{align}\label{eu}
        \mathcal E_U & =\int_{U}^{+\infty}\mathcal{Q}_{m+1}(S_{m}+x)f_{C_{m+1}}(x)\phi_{m+1}(S_{m}+x)dx \notag\\
        &\le  \int_U^{+\infty} f_{C_m}(x) dx = 1 - F_{C_m}(U),
\end{align}
where
\begin{equation}
    \begin{aligned}
        F_{C_m}(x) =
        \begin{cases}
            1 - e^{-\lambda^{-1}{\bar \gamma_m}^{-1}(2^x - 1)} & x>0\\
            0 & \rm{else}
        \end{cases}.
    \end{aligned}
\end{equation}
Given the required truncation error $\mathcal E_U$, the upper bound of $\mathcal E_U$ in \eqref{eu} indicates that the upper limit $U$ of the integral should be set as $U\geq \log_2(1-\lambda \bar \gamma_{\max} \ln \mathcal E_U)\ge \log_2(1-\lambda \bar \gamma_m \ln \mathcal E_U)$, where $\bar \gamma_{\max} = \max\{\bar \gamma_1,\cdots, \bar \gamma_M\}$. For instance, by fixing $ \bar \gamma_{\max} =20$~dB, $\lambda = 0$ dBW, and $\mathcal E_U=10^{-5}$, the minimum value of $U$ is given by $\log_2(1-1 \times 100 \times \ln 10^{-5})\approx 10.17$.

Moreover, with regard to the discretization error, it essentially stems from the trapezoidal approximation. As proved in \cite[eq.(5.1.9)]{ref37}, the asymptotic error estimate for $K\to \infty$ is given by
\begin{equation}\label{eqn:K_e}
    \mathcal E_K \approx \frac{\mathcal H^2}{12}(\psi_{m}^\prime(0) - {\psi _m^\prime} (U)).
\end{equation}
It is readily found that $\psi_{m}^\prime(x)$ is a bounded function according to the recursive relationship in \eqref{G_L} and \eqref{eqn:psi_m_x}. Therefore, the discretization error can be guaranteed by reducing the length of interval $\mathcal H$.


\subsubsection{Analysis of Computational Complexity}
The computational complexity of the trapezoidal approximation method is linearly proportional to the number of how many times of $\psi_m(S_m)$ need to be calculated, because the main computational overhead originates from the evaluation of the involved $Q $ function. Therefore, the required number of the evaluations of $\psi_m(S_m)$ in the algorithm is $(K+1)+2(K+1)+\cdots+(M-1)(K+1) = (M-1) \tilde K /2$. Accordingly, the computational complexity of trapezoidal approximation method is $O((M-1) \tilde K /2) = O(M^2K)$, where $O(\cdot)$ represents the big-O notation. According to \eqref{eu}, the computational complexity of the trapezoidal approximation algorithm can be rewritten as $O({M^2}{\log _2}(1 - \lambda {\bar \gamma_{\max}}\ln {{\cal E}_U})/{\cal H})$. Clearly, the computational complexity vastly depends on the selection of $\mathcal H$. To ensure the computational accuracy, $\mathcal H$ should be set with a small number, which consequently entails considerable complexity.

\subsection{Gauss–Laguerre Quadrature Method}
By applying Gauss–Laguerre quadrature to \eqref{phi_sm}, $\phi_{m}(S_{m})$ can be approximated as
\begin{align}\label{dp_G_L}
 &   \phi_{m}(S_{m}) =  \int\nolimits_{0}^{+\infty}e^{-x}\mathcal{Q}_{m+1}(S_{m}+x)f_{C_{m+1}}(x)\notag\\
    &\times \phi_{m+1}(S_{m}+x)e^xdx\notag\\
    \approx& \sum_{i=1}^{N} w_i
    \mathcal{Q}_{m+1}(S_{m}+\xi_i)f_{C_{m+1}}(\xi_i)\phi_{m+1}(S_{m}+\xi_i)e^{\xi_i}, 
\end{align}
where $N$ is the quadrature order (numerical experiments indicate that ${N}$ should be set at least 10), $\xi_i$ is the $i$-th root of Laguerre polynomial $\mathcal{L}_{N}(x)$, the Laguerre polynomial $\mathcal{L}_{N}(x)$ and the weights $ w_i$ are respectively given by \cite{ref38}
\begin{align}\label{cal_xi}
        \mathcal{L}_{N}(x) &= \frac{e^x}{{N}!}\frac{d^N}{dx^N}(e^{-x}x^{N}),
\end{align}
\begin{align}\label{cal_w}
        w_i = \frac{\xi_i}{({N}+1)^2[\mathcal{L}_{{N}+1}(\xi_i)]^2}.
\end{align}


By substituting \eqref{dp_G_L} into \eqref{bar_BLER_origin}, the average BLER can be approximated as
\begin{align}\label{cycle}
        \mathbb{\bar P}_M \approx& \sum_{j_1,\cdots,j_M\in [1,N]} e^{\sum_{m=1}^M\xi_{{j_m}}}\notag\\
        &\times \prod_{m=1}^M w_{j_m} f_{C_{m}}(\xi_{j_m}) \mathcal{Q}_{m}\Big(\sum\nolimits_{i=1}^{m}\xi_{j_i}\Big).
\end{align}
As aforementioned, the computational complexity is proportional to the number of evaluations of $Q$ functions. According to \eqref{cycle}, there are $MN^M$ $Q$ functions to be calculated.
Thus, the computational complexity of Gauss–Laguerre quadrature method is $O(MN^M)$. 
However, by noticing that \eqref{cycle} involves plenty of redundant calculations of the same value of $Q$ function, we resort to dynamic programming and cache the intermediate results that need to be re-calculated, eventually yields a lower computational complexity. 
\subsection{Dynamic Programming Based Method}\label{sec:dynamic}
By using the recursive relationship of $\phi_m(S_m)$ in \eqref{dp_G_L}, the average BLER can be calculated through dynamic programming. More specifically, $\phi_m(S_m)$ is iteratively updated from $m=M-1$ to $0$, where the values of $\mathcal Q_{m+1}(S_{m+1})$ and $\phi_{m+1}(S_{m+1})$ are obtained and cached to avoid redundant calculations prior to each iteration. For example, $\mathcal Q_2(S_1 + \xi_2)$ for $S_1 = \xi_1 $ is equal to $\mathcal Q_2(S_1 + \xi_1)$ for $S_1 = \xi_2 $ while applying \eqref{dp_G_L} to calculate $\phi_1(S_1)$. Likewise, there exists redundant calculations for updating $\phi_m(S_m)$. In order to conserve the computational overhead, we have to save the intermediate values of $\mathcal Q_m(S_m)$ and $\phi_m(S_m)$. Hence, dynamic programming is employed and the pseudocode of such an improved Gauss-Laguerre Quadrature method is outlined in Algorithm \ref{alg:algorithm1}.
\begin{algorithm}
    \caption{Dynamic Programming Based Gauss-Laguerre Quadrature}
    \label{alg:algorithm1}
    \begin{algorithmic}[1]
        \renewcommand{\algorithmicrequire}{\textbf{Initialization:}}
        \Require
        \State Get abscissae $\{\xi_i:\mathcal L_N(\xi_i)=0,i\in [1,N]\}$
        \State Get weights $\{w_i:i\in [1,N]\}$ via \eqref{cal_w} 
        \State Obtain sets $\mathbb S_m = \big\{S_m = \sum_{i=1}^m \xi_{j_i}: 1\leq j_{1} \le \cdots \leq j_{m}\leq N\big\}$ for $m = 1,\cdots,M$
        \State Set $S_0=0$ and $\phi_M(S_M)=1$ for $\forall S_M \in \mathbb S_M$
        \renewcommand{\algorithmicrequire}{\textbf{Dynamic Programming Enabled Calculation:}}
        \Require
        \For {$m=M-1$ to 0}
            \State Compute and cache $\mathcal Q_m(S_m)$ for $\forall S_m \in \mathbb S_m$
            \State Compute and cache $\phi_m(S_m)$ by \eqref{dp_G_L} for $\forall S_m \in \mathbb S_m$
        \EndFor
        \renewcommand{\algorithmicrequire}{\textbf{Output:}}
        \Require
        \State $\mathbb{\bar P}_M \approx \phi_0(S_0)$
    \end{algorithmic}
\end{algorithm}

In analogous to the complexity analysis of the former two numerical methods, the computational complexity of dynamic programming based algorithm is also linearly proportional to the number of evaluations of $Q$ function. The number of evaluations of $Q$ function is equal to
\begin{align}\label{dp_complexity}
       \sum_{m=1}^M{\rm card}(\mathbb S_m)
        =\sum_{m=1}^M\tbinom{m+{N}-1}{N-1}
        =\tbinom{M+{N}}{N}-1,
\end{align}
where ${\rm card}(\mathbb S)$ denotes the cardinality of set $\mathbb S$, $n\choose k$ refers to the binomial coefficient, and the first equality holds by using the combinational theory of ``stars and bars'', and the last equality holds by using \cite[eq. 0.15.1]{ref39}. 
Accordingly, the computational complexity of the dynamic programming method is expressed as $O\left(\tbinom{M+{N}}{N}\right)$.

To justify the superiority of the dynamic programming based method, we further investigate the asymptotic computational complexity under large $N$. The conclusion is drawn in the following corollary.


\begin{corollary}\label{cor1}
The computational complexity of the dynamic programming based method is asymptotically to $O(N^M/M!)$ as $N\to \infty$.
\end{corollary}
\begin{proof}
Please refer to Appendix \ref{appendices:b}.
\end{proof}

On the basis of Corollary \ref{cor1}, the computational complexity of dynamic programming based method is proved to be lower than that of Gauss–Laguerre quadrature method by around $(M+1)!$ times.

\subsection{Asymptotic Method}\label{ASYMPTOTIC ANALYSIS}
Although the average BLER can be numerically evaluated with trapezoidal approximation or Gaussian quadrature, the expression of the average BLER is too complex to extract meaningful insights and the computational complexity is also prohibitively high. To address such challenges, this section aims at the asymptotic analysis of the average BLER at high SNR, i.e., $\bar \gamma_1,\cdots,\bar \gamma_M \to \infty$. According to  \eqref{eqn:f_cm_pdf}, it follows that $f_{C_m}(x) \approx \frac{\lambda^{-1}\ln2}{\bar \gamma_m}2^{x}$ as $\bar \gamma_m\to \infty$. By applying this approximation to \eqref{E_epsilon'_m1}, the average BLER $\mathbb{\bar P}_M$ can be expressed as \eqref{E_epsilon_m11} in a recursive manner, as shown at the top of this page.

\begin{figure*}
    \begin{equation}\label{E_epsilon_m11}
     \mathbb{\bar P}_M \approx    \frac{(\lambda^{-1}\ln2)^M}{\prod_{m=1}^M\bar \gamma_m} \int\nolimits_{S_0}^{+\infty}\underbrace{ \cdots \int\nolimits_{S_{M-2}}^{+\infty}\mathcal{Q}_{M-1}(S_{M-1}) \underbrace{\int\nolimits_{S_{M-1}}^{+\infty}\mathcal{Q}_M(S_M)\underbrace{2^{S_M}}_{\psi_M(S_M)}dS_M}_{\psi_{M-1}(S_{M-1})}dS_{M-1}}_{\cdots}\cdots dS_1.
    \end{equation}
    \hrulefill
\end{figure*}


To proceed with the asymptotic analysis, the linearization approximation of $Q$ function is leveraged to simplify \eqref{E_epsilon_m11}, as shown in the following lemma.
\begin{lemma}\label{th1}
The linearization approximation of $Q$ function yields the following recursive relationship
\begin{subequations}
    \begin{align}\label{high SNR}
    \tilde f_{m}(x)=&\int_{x}^{+\infty}\mathcal{Q}_{m+1}(t)\tilde  f_{m+1}(t)dt\\
    \approx& \frac{1}{2{V}_{m+1}}\int_{R-{V}_{m+1}}^{R+{V}_{m+1}}\tilde F_{m+1}(t)dt-\tilde F_{m+1}(x),\label{high SNR1}
    \end{align}
\end{subequations}
where $\tilde F_{m+1}(x) = \int_{0}^{x} \tilde f_{m+1}(t)dt$ and ${V}_{m} = \sqrt{\frac{m\pi}{2L}}\log_2 e$. 
\end{lemma}
\begin{proof}[Proof]
Please refer to Appendix \ref{appendices:a}. 
\end{proof}
On the basis of Lemma \ref{th1} and $\psi_{M}(S_{M})=2^{S_M}$, $\psi_{M-1}(S_{M-1})$ can be obtained as
\begin{align}
    \psi_{M-1}(S_{M-1}) &= \int_{S_{M-1}}^{+\infty}\mathcal{Q}_{m+1}(S_M)\psi_M(S_M)dt\notag\\
    &=\frac{2^{R+V_M}-2^{R-V_M}}{(\ln 2)^2} - \frac{2^{S_{M-1}}}{\ln2}.
\end{align}
Accordingly, it can be proved by induction that $\psi_m(S_m)$ can be expressed in a general form as given by the following theorem.
\begin{theorem}\label{the:phi_m}
$\psi_{m}(S_m)$ is obtained by
  \begin{equation}
        \psi_{m}(S_m) = \sum_{i=0}^{M-m}\alpha_{m,i}{(S_m)}^{i} + (-\ln2)^{m-M}2^{S_m},\, m\ge 0,
\end{equation}
where $\alpha_{M,0}=0$, $\alpha_{m,0}$ for $0\leq m<M$ is given by \eqref{eqn:alpha_def}, as shown at the top of the next page,
\begin{figure*}
   \begin{align}\label{eqn:alpha_def}
    \alpha_{m,0} = \frac{1}{2V_{m+1}}\bigg(-(-\ln2)^{m-M}(2^{R+ V_{m+1}}-2^{R-V_{m+1}})+\sum\nolimits_{i=0}^{M-m-1}\alpha_{m+1,i}\frac{(R+V_{m+1})^{i+1}-(R-V_{m+1})^{i+1}}{i+1}\bigg),
\end{align}
    \hrulefill
\end{figure*}
and $\alpha_{m,i}=-\alpha_{m+1,i-1}/i$ if $1 \leq i \leq M-m $. 
\end{theorem}
\begin{proof}
  With Lemma \ref{th1}, \eqref{eqn:alpha_def}  can be proved by induction. The detailed proof is omitted here to save space.
\end{proof}

By applying Theorem \ref{the:phi_m} to \eqref{E_epsilon_m11}, the average BLER at high SNR can be derived as
\begin{align}\label{Asy Con}
    \mathbb{\bar P}_M
    \simeq& \frac{(\lambda^{-1}\ln2)^M}{\prod_{m=1}^M\bar \gamma_m} \psi_{0}(S_0)\notag\\
    =&\frac{(\lambda^{-1}\ln2)^M}{\prod_{m=1}^M\bar \gamma_m}\Big(\alpha_{0,0} + (-\ln2)^{-M}\Big)\notag\\
    =&\mathcal{G}_M\prod\nolimits_{m=1}^{M}\bar \gamma_m^{-1},
\end{align}
where $\mathcal{G}_M= \lambda^{-M}\left((\ln2)^M\alpha_{0,0}+(-1)^M\right)$. Since it is unnecessary to evaluate $Q$ function, the computational complexity of the asymptotic method is $O(1)$. Besides, the simple form of the asymptotic average BLER significantly facilitates the optimal design of HARQ-IR assisted short packet communications.

\section{Maximization of LTAT}\label{Throughput Optimization}

Due to the restricted battery capacity at the mobile terminals, a proper power allocation scheme is of necessity to realize green communications. 
In this section, we focus on the power allocation across different HARQ rounds for HARQ-IR aided short packet communications by solely relying on the statistical information of channel states. More specifically, the LTAT is maximized while ensuring low average power consumption and stringent average BLER constraint. In the following, the problem of LTAT maximization is first formulated and two different methods are developed to solve it.


\subsection{Problem Formulation}
The LTAT $\eta$ is defined as the average throughput in a long time period precisely given by \cite{ref42} 
\begin{equation}
    \begin{aligned}\label{Generally Throughput}
        \eta = \frac{R(1 - \mathbb{\bar P}_M)}{1 + \sum_{m=1}^{M-1}\mathbb{\bar P}_m},
    \end{aligned}
\end{equation}
Accordingly, the LTAT is maximized through the optimal power allocation by considering the constraints of the average transmission power and the average BLER, which mathematically reads as
\begin{align}\label{throughput_maximization}
    \max_{P_1,\cdots,P_M}&\quad \eta \notag\\
    \rm{s.t.}&\quad P_{avg} \leq \bar P \notag\\
    &\quad \mathbb{\bar P}_M \leq \mathbb{\bar P},
\end{align}
where $P_{avg} = \sum_{m=1}^M P_m \mathbb{\bar P}_{m-1}$ denotes the average total transmission power for delivering each message, $\bar P$ and $\mathbb{\bar P}$ denote the maximum allowable average power and average BLER, respectively. Due to the fractional objective function and non-convex feasible set, \eqref{throughput_maximization} cannot be directly solved. To address this issue, two different methods are developed in the following to transform \eqref{throughput_maximization} into a convex problem and a constrained MDP, which can be solved with geometric programming (GP) and deep reinforcement learning (DRL), respectively.
\subsection{GP-Based Method}

By noticing the simplicity of the asymptotic BLER in \eqref{Asy Con}, the asymptotic BLER is leveraged to convert \eqref{throughput_maximization} into a fractional programming problem which can be solved with GP-based method. More specifically, by taking into account that the maximum tolerable BLER after $M$ rounds is extremely low in practice, i.e., $\mathbb {\bar P}_M \ll 1$, the LTAT can be approximated as  $\eta \approx {R}/({1 + \sum_{m=1}^{M-1}\mathbb{\bar P}_m})$. Then we replace the average BLER $\mathbb {\bar P}_M$ in \eqref{throughput_maximization} with the asymptotic expression of BLER in \eqref{Asy Con}, we thus have
\begin{equation}
    \begin{aligned}
        \begin{split}\label{Optimization problem}
            \max_{P_1,\cdots P_M}&\quad \frac{R}{1+\sum_{m=1}^{M-1}\mathcal{G}_{m}\prod_{i=1}^{m}P_i^{-1}}\\
            \rm{s.t.}&\quad \sum_{m=1}^{M}P_m\mathcal G_{m-1}\prod_{i=1}^{m-1}P_i^{-1}\leq \bar P\\
            &\quad \mathcal{G}_M\prod_{m=1}^MP_m^{-1}\leq \mathbb{\bar P},
        \end{split}
    \end{aligned}
\end{equation}
where $\mathcal G_0 = 1$ and the average power of noise is normalized to unity, i.e., $\mathcal N_0 = 1$, for notational simplicity.

Since the numerator of the above objective function is independent of $R$, \eqref{Optimization problem} is equivalent to the following minimization problem as
\begin{equation}\label{Optimizationproblem11}
    \begin{aligned}
        \begin{split}
            \min_{P_1,\cdots P_M}&\quad \sum_{m=1}^{M-1}\mathcal{G}_m\prod_{i=1}^m P_i^{-1}\\
            \rm{s.t.}&\quad \sum_{m=1}^{M}\mathcal G_{m-1} P_m\prod_{i=1}^{m-1}P_i^{-1}  \leq \bar P \\
            &\quad \mathcal{G}_M \prod_{m=1}^MP_m^{-1} \leq \mathbb{\bar P},
        \end{split}
    \end{aligned}
\end{equation}

Clearly, \eqref{Optimizationproblem11} is a GP problem that can be easily converted into a convex problem. There are plenty of readily available software packages for solving GP problem, such as Matlab toolbox ``GGPLAB''. The globally optimal solution to \eqref{Optimizationproblem11} can be numerically computed.

\subsection{DRL-Based Method}
However, by noticing that the asymptotic results are only valid at high SNR, the asymptotic BLER is greatly larger than the exact value. Hence, the optimal solution offered by GP method significantly underestimates the throughput performance. To combat this shortcoming, the optimization problem is formulated as a constrained MDP which can be solved with DRL-based method. In doing so,
we first transform \eqref{throughput_maximization} into an equivalent unconstrained Lagrange duality problem \cite{ref41}
\begin{align}\label{Lagrange_problem}
    \min_{\rho, \nu} \max_{P_1,\cdots,P_M} \ell,
\end{align}
where the dual variables $\rho,~ \nu\geq 0$, and
\begin{align}\label{Lagrange}
    \ell = &\eta + \rho(\bar P - P_{avg})+\nu(\mathbb{\bar P} - \mathbb{\bar P}_M).
\end{align}
To proceed with the optimization, the dual problem \eqref{Lagrange_problem} is reformulated as a MDP problem.

\subsubsection{Problem Reformulation}
To enable the reformulation of \eqref{Lagrange_problem} as a MDP, each term in the objective function of \eqref{Lagrange_problem} should be written as an accumulative reward after $T$ time steps.
To this end, we assume HARQ transmission procedure to be an ergodic processes similarly to \cite{ref43}. With regard to $\eta$, the LTAT defined in \eqref{Generally Throughput} can be equivalently expressed as
\begin{equation}
    \begin{aligned}\label{LATA}
        \eta = \mathbb{E}\left[\frac{1}{T}\sum_{t=1}^T \mathcal R(t)\right],
    \end{aligned}
\end{equation}
where $T$ denotes the total number of HARQ rounds and $\mathcal R(t)$ is the effective transmission rate in the $t$-th HARQ round. $\mathcal R(t) = R$ if the message is successfully recovered, and $\mathcal R(t) = 0$ otherwise. Moreover, with regard to the average BLER after $ m$ rounds, $\bar{\mathbb{P}}_m$ can be rewritten as
\begin{align}\label{LTA-BLER}
    \bar{\mathbb{P}}_m & = \frac{{\rm The~number~of~failed~messages~after~}m\rm ~rounds}{\rm The~number~of~delivered~messages} \notag\\
    &= \mathbb{E}\left[\frac{\sum_{t=1}^T\mathbb{I}\left(\mathcal{F}_{\frak n(t)=m}\right)}{T/\bar M} \right],\,m \ge 1,
\end{align}
where $\mathbb{I}(\cdot)$ refers to the indicator function, $\frak n(t)$ denotes the mapping from the time slot $t$ to the current HARQ round, 
$\mathcal{F}_{\frak n(t)=m}$ denotes the occurrence of failure at time slot $t$ together with $\frak n(t)=m$, and $\bar M$ stands for the average number of HARQ rounds for delivering one message.

Furthermore, with regard to $P_{avg}$, the total average transmission power $P_{avg}$ can be expressed in terms of time-average as
\begin{align}\label{LTA-avg}
    P_{avg} & = \frac{{\rm The~total~transmission~power}}{\rm The~number~of~delivered~messages} \notag\\
    &= \mathbb{E}\left[\frac{\sum_{t=1}^T \mathcal P(t)}{T/\bar M} \right],
\end{align}
where $\mathcal P(t)$ represents the transmission power in time slot $t$.

By substituting  \eqref{LATA} and \eqref{LTA-BLER} into \eqref{Lagrange}, the objective function $\ell$ can be written as \eqref{MDP-Subgradient11}, as shown at the top of this page. 
\begin{figure*}
    \begin{align}\label{MDP-Subgradient11}
    \ell =&\mathbb{E}\left[\frac{1}{T}\sum_{t=1}^T \mathcal R(t) + \rho\Bigg(\bar P - \frac{1}{T}\sum_{t=1}^T \bar M\mathcal P(t) \Bigg)+ \nu\Bigg(\mathbb{\bar P}- \frac{1}{T}\sum_{t=1}^T\bar M\mathbb{I}\left(\mathcal{F}_{\frak n(t)=M}\right)\Bigg)\right]\notag\\
    =&\mathbb{E}\left[\frac{1}{T}\sum_{t=1}^T\left( \mathcal R(t)+ \rho\left(\bar P - \bar M\mathcal P(t)\right) + \nu\left(\mathbb{\bar P} - \bar M\mathbb{I}\left(\mathcal{F}_{\frak n(t)=M}\right)\right)\right)\right].
\end{align}
\hrulefill
\end{figure*}
With the reformulated objective function $\ell$, the power allocation policy of the inner optimization problem in \eqref{MDP-Subgradient11} can be modeled as an MDP problem which can be solved by invoking deep reinforcement learning (DRL) method.

To be more specific, a MDP is fully characterized by four elements, including environment, state space $\mathbb S$, action space $\mathbb A$, and reward space $\mathbb R$. According to the mechanism of HARQ-IR-aided short packet communications, the transmitter plays as an agent to select a proper transmission power for each time slot. It is noteworthy that the decision of allocated power is made according to the current state $\bm s_t$. After taking this decision, a new state $\bm s_{t+1}$ will be observed together with a reward $r_t$ returned by the environment.
To reformulate the power allocation problem of HARQ-IR-aided short packet communications as an MDP, the states, actions, and rewards are devised as follows:
\begin{itemize}
\item \textbf{Action} ${a}_t$: The action is defined as the transmission power ${a}_t = P(t) \in \mathbb A^{1\times 1}$ for the current HARQ round according to the current state ${s}_t$.

\item \textbf{State} ${s}_t$: 
Since the selection of the transmission power for state $s_t$ is related to the allocated powers in the prior HARQ rounds, the state $s_t$ is defined as 
${s}_t = (\mathcal P({t+1-\frak n(t)}),\cdots, \mathcal P({t-1}),0,\cdots,0) \in \mathbb S^{1\times (M-1)}$.  Moreover, if $\frak n(t)=1$, we stipulate ${s}_t = (0,0,\cdots,0)$. Besides, it is worth emphasizing that only statistical channel state information (CSI) is available at the transmitter, instantaneous CSI cannot be acquired prior to transmissions. Accordingly, $s_t$ does not include channel states.

\item \textbf{Reward} ${r}_t$: After executing action ${a}_t$, a reward ${r}_t \in \mathbb R^{1\times 1}$ will be received from the environment. According to \eqref{MDP-Subgradient11}, the reward ${r}_t$ can be set as
\begin{align}\label{reward}
    {r}_t = 
    \mathcal R(t)+ \rho(\bar P - \bar M\mathcal P(t)) + \nu\left(\mathbb{\bar P} - \bar M\mathbb{I}\left(\mathcal{F}_{\frak n(t)=M}\right)\right).
\end{align}
Clearly, the value of $\bar M$ depends on the power allocation policy. Accordingly, $\bar M$ should be updated periodically during training. More detailed discussions are deferred to Subsection \ref{sec:update}.


\end{itemize}

\begin{figure*}[h]
\centerline{\includegraphics[width=1\linewidth]{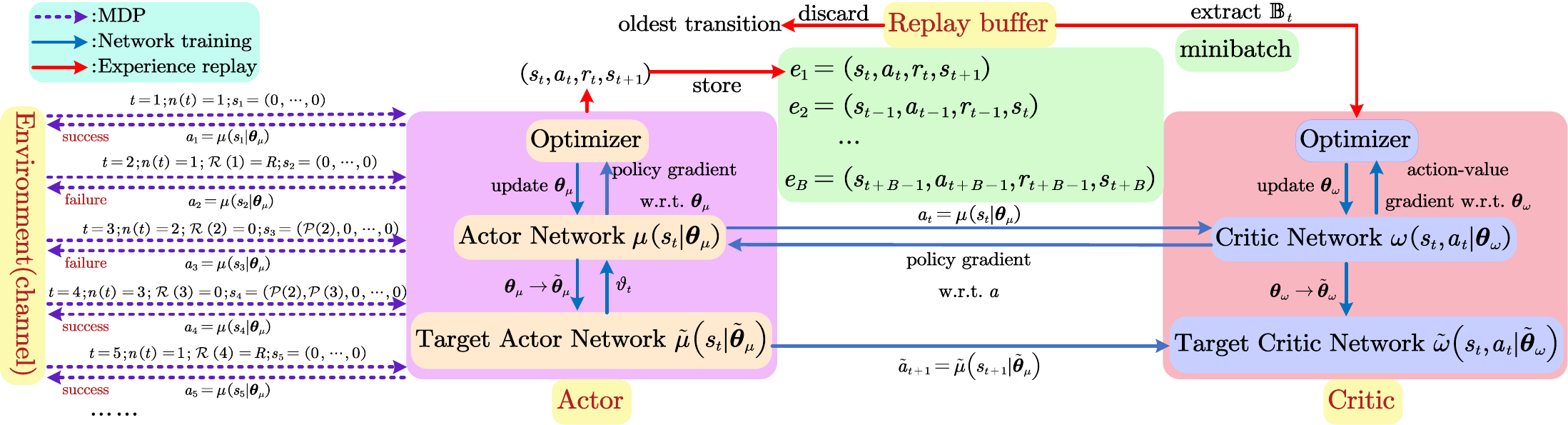}}
\caption{The DDPG network for power allocation of HARQ-IR-aided short packet communications.}
\label{DRL_process}
\end{figure*}
\subsubsection{DDPG}
By considering the continuity of the action space $\mathbb A$ and the state space $\mathbb S$, a DRL algorithm based on DDPG proposed in \cite{ref14} can be employed to solve the formulated MDP problem. More specifically, the DDPG algorithm is essentially developed based on the actor-critic framework, which comprises a actor network $\mu(s_t|{\bm \theta}_\mu)$, a critic network $\omega(s_t,a_t|{\bm \theta}_\omega)$, two target networks, and a experience replay buffer. For illustration, the DDPG-based power allocation scheme for HARQ-IR-aided short packet communications is shown in Fig. \ref{DRL_process}.
In Fig. \ref{DRL_process}, the actor network $\mu(s_t|{\bm \theta}_\mu)$, parameterized by ${\bm \theta}_\mu$, is used to output an action $a_t$ according to the current state $s_t$, and the critic network $\omega(s_t,a_t|{\bm \theta}_\omega)$, parameterized by ${\bm \theta}_\omega$, is adopted to estimate the Q-values for the output state-action pair $(s_t,a_t)$. In addition, the two target networks consists of a target actor network $\tilde \mu(s_t|\tilde {\bm \theta}_\mu)$ and a target critic network $\tilde \omega(s_t,a_t|\tilde {\bm \theta}_\omega)$, and both target networks are employed to stabilize the learning procedure by offering consistent target values. Besides, the replay buffer is leveraged to store experiences for training. In the sequel, the sampling strategy of replay buffer and the parameter training of the four networks are described in detail.



\begin{itemize}
\item \textbf{Replay Buffer}: In each time slot $t$, the current state $s_t$, the chosen action $a_t$, the received reward $r_t$, and the next state $s_{t+1}$ constitute a transition $e_t = (s_t,a_t,r_t,s_{t+1})$ that will be stored as an experience in the replay buffer with a finite size $B$. Moreover, the oldest transition is discarded and the currently generated transition is added when the replay buffer is full. In order to stabilize as well as accelerate the learning process, we capitalize on the prioritized sampling strategy for the training of the four networks \cite{ref44}.
More specifically, a minibatch of $\mathcal B$ transitions are extracted from the prioritized replay buffer in each training step, wherein each transition $e_i = (s_i,a_i,r_i,s_{i+1})$ is sampled with a probability $p_i$. The sampling probability $p_i$ is precisely given by \cite{ref44}
\begin{equation}
    p_i \propto |\vartheta_i| + \varepsilon,
\end{equation}
where $\varepsilon$ is positive to avoid a null sampling probability, $\vartheta_i$ refers to the temporal difference (TD) error given by
\begin{align}
      \vartheta_i = \omega(s_{i},a_i|{\bm \theta}_{\omega}) -r_i - \zeta \tilde\omega(s_{i+1}, \tilde{\mu}(s_{i+1}|\tilde{\bm \theta}_\mu)|\tilde{\bm \theta}_\omega),
\end{align}
and $\zeta$ denotes discount factor.


\item \textbf{Critic Network}: 
The parameters $\bm\theta_\omega$ of the critic network $\omega(s_t,a_t|{\bm \theta}_\omega)$ can be updated by the TD algorithm. Specifically, 
the loss function of training $\omega(s_t,a_t|{\bm \theta}_\omega)$ is denoted by the weighted squared TD error averaging over the sampled minibatch $\mathbb B_t$, i.e.,
    \begin{align}
      \pounds({\bm \theta}_\omega) = \frac{1}{2\mathcal{B}}\sum_{e_i\in\mathbb {B}_t}\varphi_i\vartheta_i^2,
    \end{align}
    where $\varphi_i\propto(\mathcal B p_i)^{-\beta}$ is the importance-sampling weight and $\beta\in[0,1]$ \cite{ref44}. Accordingly, by invoking gradient descent algorithm, ${\bm \theta}_\omega$ can be updated as
    \begin{align}
      \bm \theta_{\omega}^{\rm new} = \bm \theta_{\omega}^{\rm now} - \tau_{\omega}\nabla_{\bm \theta_{\omega}} \pounds({\bm \theta}_\omega^{\rm now}),
    \end{align}
    where $\tau_{\omega}$ is the learning rate and the gradient of loss function $\nabla_{\bm \theta_{\omega}}\pounds({\bm \theta}_\omega)$ is
    \begin{align}\label{theta_omega1}
      \nabla_{\bm \theta_{\omega}} \pounds({\bm \theta}_\omega) = \frac{1}{\mathcal B}\sum_{e_i \in \mathbb{B}_t}\varphi_i\vartheta_i\nabla_{\bm{\theta}_\omega}\omega (s_i,a_i|{\bm \theta}_\omega).
    \end{align}
\item \textbf{Actor Network}: 
The parameters ${\bm \theta}_\mu$ of the actor network $\mu(\cdot|{\bm \theta}_\mu)$ are similarly updated by using gradient descent algorithm according to the score of the current action policy, i.e., $J(\bm{\theta}_\mu)$. More precisely, $J(\bm{\theta}_\mu)$ is given by \cite{ref45} 
     \begin{align}
        J(\bm{\theta}_\mu) = \frac{1}{\mathcal B}\sum_{e_i\in \mathbb B_t} \omega(s_i, \mu(s_i|{\bm \theta}_\mu)|{\bm \theta}_{\omega}).
    \end{align}
It is noteworthy that a higher score $J(\bm{\theta}_\mu)$ of action indicates a better action policy. To achieve the best action policy, the parameters ${\bm \theta}_\mu$ must be optimized through maximizing $J(\bm{\theta}_\mu)$. Therefore, applying the gradient ascendant method yields the updated ${\bm \theta}_{\mu}$ as
    \begin{align}\label{theta_mu1}
        \bm \theta_{\mu}^{\rm new} = \bm \theta_{\mu}^{\rm now} + \tau_{\mu}\nabla_{\bm \theta_{\mu}} J({\bm \theta}_\mu^{\rm now}),
    \end{align}
    where $\tau_\mu$ is the learning rate and $\nabla_{\bm{\theta}_\mu} J(\bm{\theta}_\mu)$ is given by \cite{ref14}
    \begin{align}
        \nabla_{\bm{\theta}_\mu} J(\bm{\theta}_\mu) = \frac{1}{\mathcal{B}}\sum_{e_i\in\mathbb B_t}\nabla_{{\bm \theta}_{\mu}}\mu(s_i|{\bm \theta}_{\mu}) \nabla_{a}\omega(s_i,\mu(s_i|{\bm \theta}_{\mu})|{\bm \theta}_{\omega}).
    \end{align}
\item \textbf{Target Networks}: To ensure the stability of updates, the soft update strategy is leveraged to update both target networks. Therefore, the network parameters $\tilde{\bm \theta}_\omega$ and $ \tilde{\bm \theta}_\mu$ can be respectively updated by
    \begin{align}\label{theta_omega2}
    \tilde{\bm \theta}_\omega^{\rm new} =& \varrho \bm \theta_\omega^{\rm new} + (1-\varrho) \tilde{\bm \theta}_\omega^{\rm now},
    \end{align}
    \begin{align}\label{theta_mu2}
    \tilde{\bm \theta}_\mu^{\rm new} =& \varrho \bm \theta_\mu^{\rm new} + (1-\varrho) \tilde{\bm \theta}_\mu^{\rm now},
    \end{align}
    where $\varrho \ll 1$ is a hyperparameter.
\end{itemize}

\subsubsection{Updates of $\rho$, $\nu$, and $\bar M$}\label{sec:update}
With regard to the outer minimization problem in \eqref{MDP-Subgradient11}, the subgradient method \cite{ref45} is applied to update the dual variables $\rho$ and $\nu$. More specifically, $\rho$ and $\nu$ can be updated respectively by
\begin{align}\label{rho update}
    \rho^{\rm new} = \left[\rho^{\rm now}-\tau_\rho \left(\bar P - \frac{\sum_{t=1}^T \mathcal P(t)}{T/\bar M}\right)\right]^+,
\end{align}
\begin{align}\label{nu update}
    \nu^{\rm new} = \left[\nu^{\rm now} - \tau_\nu \left(\mathbb{\bar P}-\frac{\sum_{t=1}^{T}\mathbb{I}(\mathcal F_{\frak n(t)=M})}{T/\bar M}\right)\right]^+,
\end{align}
where $[x]^+=\max\{x,0\}$, $\tau_\rho $ and $ \tau_\nu$ represent the step sizes.

Besides, it is worthwhile to highlight that the average number of transmissions, i.e., $\bar M$, in \eqref{reward}, \eqref{rho update}, and \eqref{nu update} should be periodically updated to continuously approach the actual value of $\bar M$ for the current power allocation policy.
According to the definition of the average number of HARQ rounds in \cite{ref42} together with \eqref{LTA-BLER}, $\bar M$ can be expressed as
\begin{align}\label{eqn:M_bar_def}
    \bar M =& 1 + \sum_{m=1}^{M-1} \bar{\mathbb{P}}_m
    = 1 + \frac{\sum_{t=1}^{T}\mathbb I(\mathcal F_{\frak n(t)<M})}{T/\bar M},
\end{align}
where $\mathbb I(\mathcal F_{\frak n(t)<M})$ denotes the occurrence of decoding failures in the first $M-1$ HARQ rounds.
On the basis of \eqref{eqn:M_bar_def}, $\bar M$ can be updated as
\begin{equation}
    \begin{aligned}\label{bar_M}
        \bar M = \frac{T}{T-\sum_{t=1}^{T}\mathbb I(\mathcal F_{\frak n(t)<M})}.
    \end{aligned}
\end{equation}

\subsubsection{Truncation-Based Updating Strategy}
In order to ensure the convergence and stability of the training process, a truncation-based updating approach is proposed to update $\rho^{\rm new}$, $\nu^{\rm new}$, and $\bar M$. More specifically, the summations over $T$ time slots involved in \eqref{rho update}, \eqref{nu update}, and \eqref{bar_M} are truncated to the last $W$ time slots. Hence, $\rho^{\rm new}$, $\nu^{\rm new}$, and $\bar M$ are respectively rewritten as
\begin{align}\label{rho update1}
    \rho^{\rm new} = \left[\rho^{\rm now}-\tau_\rho \left(\bar P - \frac{\sum_{t=T-W+1}^T \mathcal P(t)}{W/\bar M}\right)\right]^+,
\end{align}
\begin{align}\label{nu update2}
    \nu^{\rm new} = \left[\nu^{\rm now} - \tau_\nu \left(\mathbb{\bar P}-\frac{\sum_{t=T-W+1}^{T}\mathbb{I}(\mathcal F_{\frak n(t)=M})}{W/\bar M}\right)\right]^+,
\end{align}
\begin{equation}
    \begin{aligned}\label{bar_M3}
        \bar M = \frac{W}{W-\sum_{t=T-W+1}^{T}\mathbb I(\mathcal F_{\frak n(t)<M})}.
    \end{aligned}
\end{equation}
The whole pseudocode of the DDPG-based algorithm is outlined in Algorithm \ref{alg:algorithm2}.

\begin{algorithm}
    \caption{The DDPG-based power allocation scheme}
    \label{alg:algorithm2}
    \begin{algorithmic}[1]
        \State Initialize $\rho$, $\nu$, $W$, $B$, $\bm \theta_\mu$,  and $ \bm \theta_\omega$
        \State Set $\tilde {\bm\theta}_\mu = \bm\theta_\mu$ and $\tilde{\bm \theta}_\omega = \bm\theta_{\omega}$
        \Repeat
        \State Choose action $a_t = \mu(s_t|\theta_\mu)$ based on $s_t$
        \State Observe $\mathcal R(t)$ and $s_{t+1}$ after executing $a_t$
        \State Update $\bar M$ by \eqref{bar_M3} 
        \State Calculate the reward $r_t$ by \eqref{reward}
        \State Store transition $({s}_t, {a}_t, {r}_t, {s}_{t+1})$ in replay buffer
        \State Sample a minibatch $\mathbb B_t$ from  prioritized replay buffer
        \State Update $\theta_\omega$, $\theta_\mu$, $\tilde{\theta}_\omega$, and $\tilde{\theta}_\mu$ by
        \eqref{theta_omega1}, \eqref{theta_mu1}, \eqref{theta_omega2} and \eqref{theta_mu2}, respectively
        \If {$t~{\rm mod}~I  = 0$ } 
            \State Update $\rho$ and $\nu$ by \eqref{rho update1} and \eqref{nu update2}, respectively
        \EndIf
        \Until convergence
    \end{algorithmic}
\end{algorithm}

\section{Numerical Analysis}\label{SIMULATION ANALYSIS}
This section is devoted to verify our analytical results through extensive simulations. For illustration, we assume that the average transmit SNRs in all HARQ rounds keep fixed, i.e., $\bar \gamma_1 = \cdots = \bar \gamma_M = \bar \gamma$. Unless otherwise stated, the system parameters are set as $R = 5$ bps/Hz, $ L = 50$ symbols, $\lambda = 1$, $M=5$, $N = 20$, $K = 3000$, $U = 10$, $B=3000$, $W=300$, and $I=100$.
\subsection{Verification of BLER Analysis}
In Figs. \ref{fig1}-\ref{fig3}, the numerical results via Monte Carlo simulation, trapezoidal approximation method, Gauss–Laguerre quadrature, and the asymptotic method are respectively obtained according to \eqref{origin_real}, Subsection \ref{sec:trap}, Subsection \ref{sec:dynamic}, and \eqref{Asy Con}. These numerical results are labeled as ``Sim.'', ``Trap.'', ``Gauss.'', ``Asy.'', respectively. In Fig. \ref{fig1}, the average BLER is plotted against the average transmit SNR by considering different $M$. It can be observed that the simulation results coincide with the ones through both trapezoidal approximation method and Gauss-Laguerre quadrature. This justifies the validity of our analytical results. Moreover, it can be observed from Fig. \ref{fig1} that the asymptotic results tend to the simulation ones as $\bar \gamma$ increases.


\begin{figure}[htbp]
\centering
\includegraphics{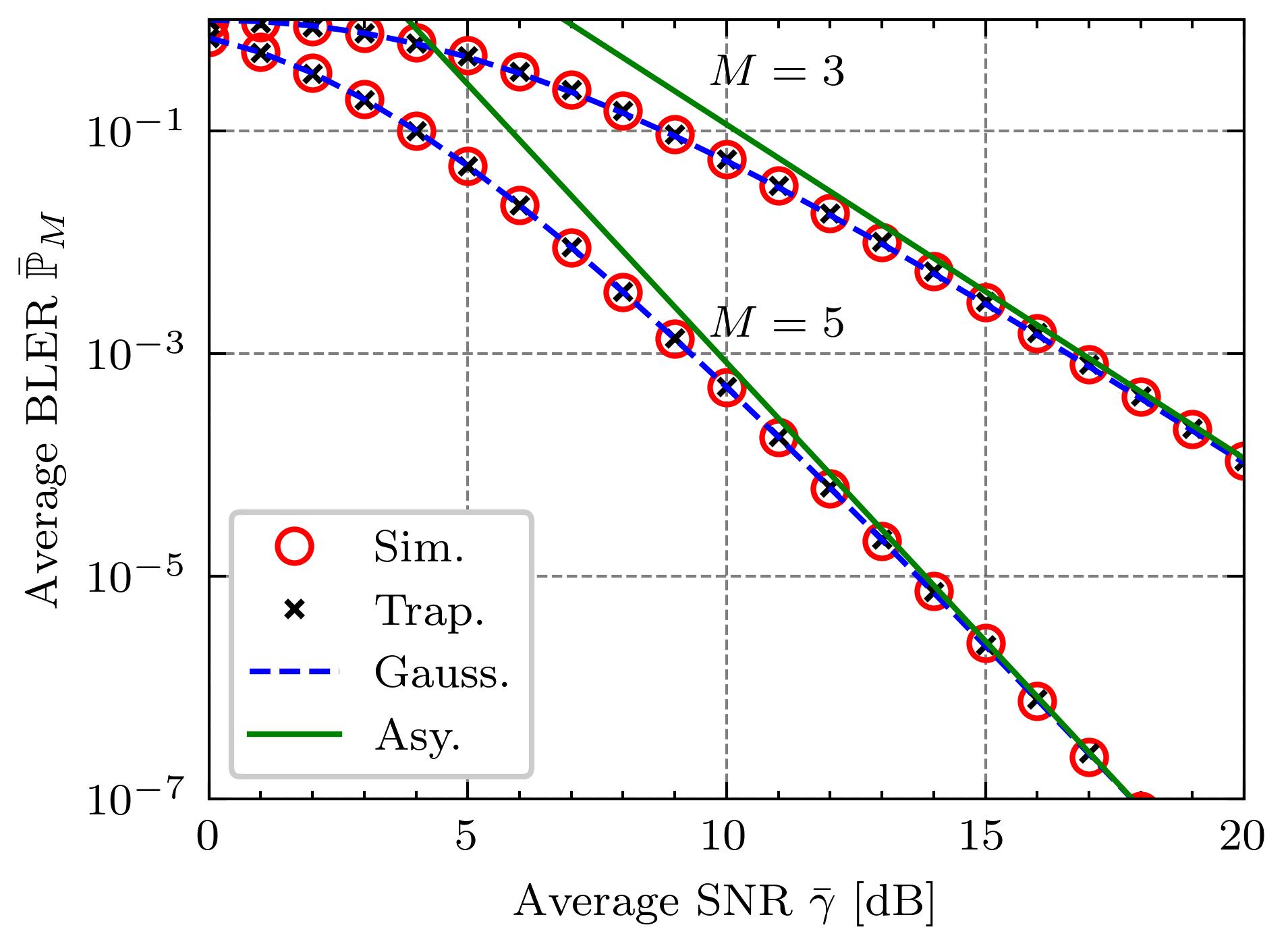}
\caption{The average BLER $\mathbb {\bar P}_M$ versus the average SNR $\bar \gamma$.}
\label{fig1}
\end{figure}


Fig. \ref{fig2} depicts the average BLER versus the maximum number of transmissions under different $\bar \gamma$. Likewise, the simulation results, trapezoidal approximation method, and Gauss–Laguerre quadrature agree well with each other. In addition, the gap between the simulation and the asymptotic results diminish as $\bar \gamma $ increases.


\begin{figure}[htbp]
\centering
\includegraphics{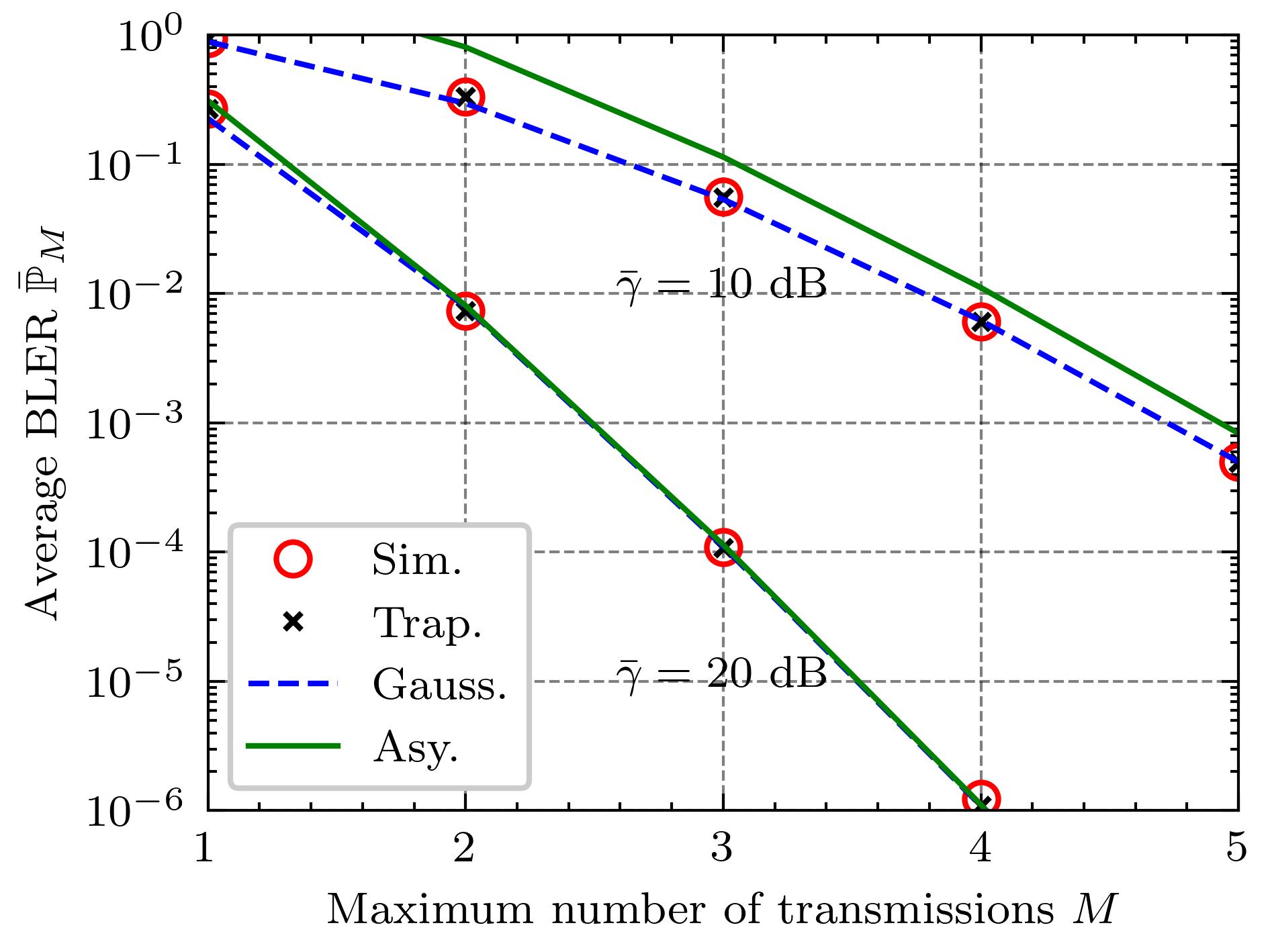}
\caption{The average BLER $\mathbb {\bar P}_M$ versus the maximum number of transmissions $M$.} 
\label{fig2}
\end{figure}

Fig. \ref{fig3} illustrates the impact of the length of the sub-codeword upon the average BLER $\mathbb P_M$ given $\mathcal K = 1000$ bits, where the transmission rate $R=\mathcal K/L$. It can be observed that the simulation and the analytical results are in perfect match. Moreover, the gap between the asymptotic and the simulation results becomes small as $\bar \gamma$ increases or $L$ increases.


\begin{figure}[htbp]

\centering
\includegraphics{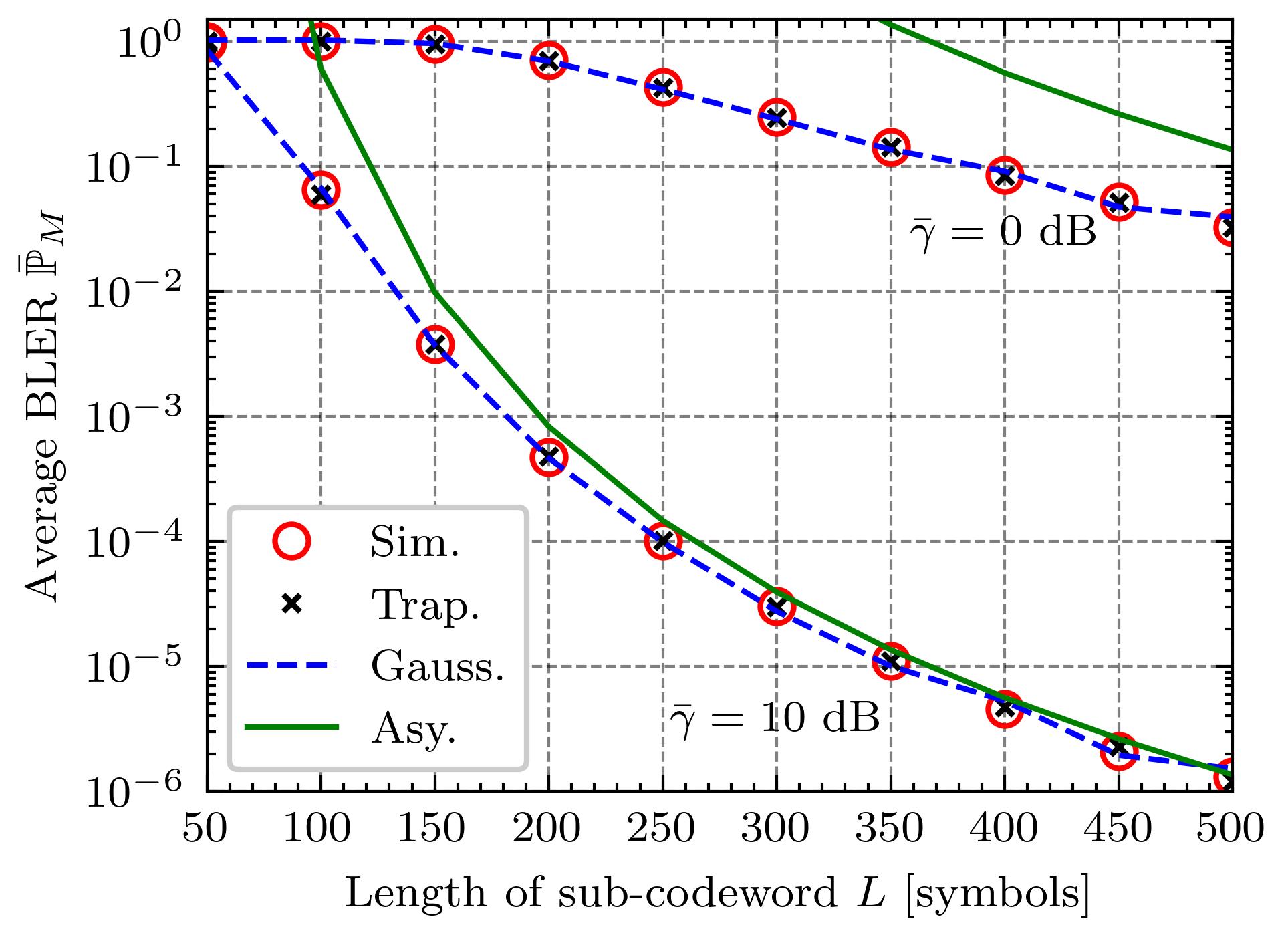}
\caption{The average BLER $\mathbb {\bar P}_M$ versus the length  $L$ of sub-codeword.}
\label{fig3}
\end{figure}

\subsection{Maximization of LTAT}
In Fig. \ref{Throughput}, the maximum LTAT is plotted against the maximum allowable average transmission power $\bar P$ for GP-based and DRL-based methods by setting $\mathbb{\bar P} = 0.01$. The numerical results reveal that the DRL-based method performs better than the GP-based one under stringent power constraint. This is because of the fact that the GP-based method is based on the asymptotic expression of the average BLER, which exhibits a high approximation error at low SNR. As a consequence, the GP-based method significantly underestimates the system performance under low SNR. This justifies the great value of the DRL-based method in designing the power allocation scheme. However, under a sufficiently loose power constraint, such as $\bar P> 25$ dB, the DRL-based method can attain almost a comparable performance as the GP-based one. Nevertheless, regarding the power allocation at high SNR, the GP-based method still constitutes an much appealing low-complexity algorithm by comparing to the DRL-based one.



\begin{figure}[htbp]
\centering
\includegraphics{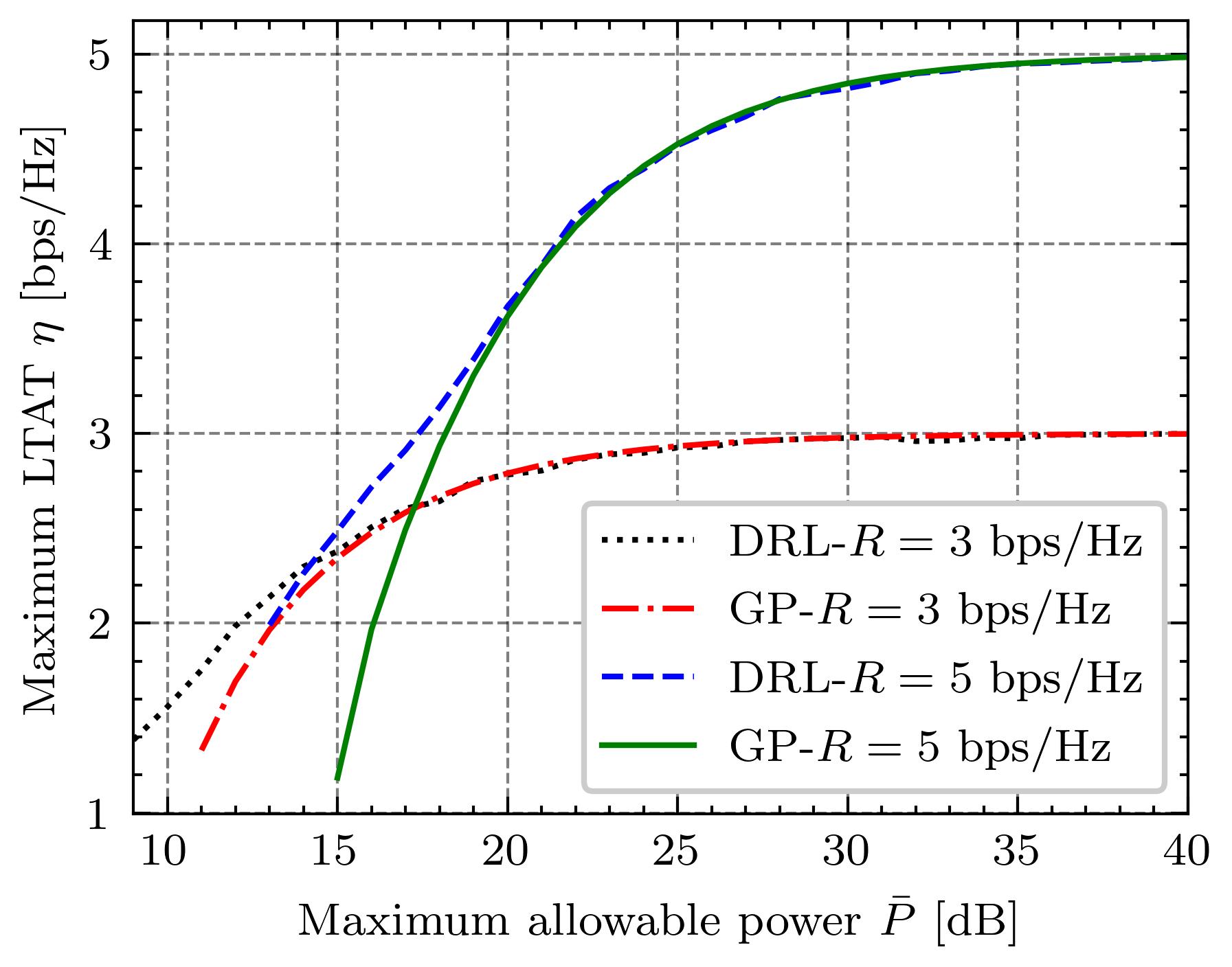}
\caption{The maximum LTAT $\eta$ versus the maximum allowable power $\bar P$.}
\label{Throughput}
\end{figure}


\subsection{The Convergence of the DDPG-Based Algorithm}
Figs. \ref{LTAT_conv}, \ref{LTA_BLER}, and \ref{LTA_Power} investigate the convergence and the stability of the DDPG-based algorithm, wherein $\bar P/\mathcal N_0 = 20$ dB, $\mathbb{\bar P} = 0.01$, $R = 5$ bps/Hz, $ L = 50$ symbols, and $M = 5$. 
Moreover, the labels ``Trun.'', ``No Trun.'', and ``Bound'' in the following three figures stand for the proposed DDPG-enabled power allocation with truncation-based updating, without truncation updating, and the upper bound of the average BLER/power, respectively. Fig. \ref{LTAT_conv} shows the convergence of the LTAT with regard to the number of training epoches for different truncation length $W$. It is observed that the proposed truncation-based updating strategy converges faster than the one without truncation. Besides, the truncation-based updating strategy achieves a slightly greater LTAT than the one without truncation. However, it is worth emphasizing that the truncation length $W$ as a hyperparameter should be properly chosen according to the size of replay buffer $B$, the updating period $I$ of dual variables, the learning rate, etc. For instance, it can be observed from Fig. \ref{LTAT_conv} that the convergence of the DRL-based method becomes unstable if $W$ is too small (say $W=10$). To ensure the stability as well as the performance, we set $W=300$ in Figs. \ref{LTA_BLER} and \ref{LTA_Power}.

\begin{figure}[htbp]
\centering
\includegraphics[width=3in]{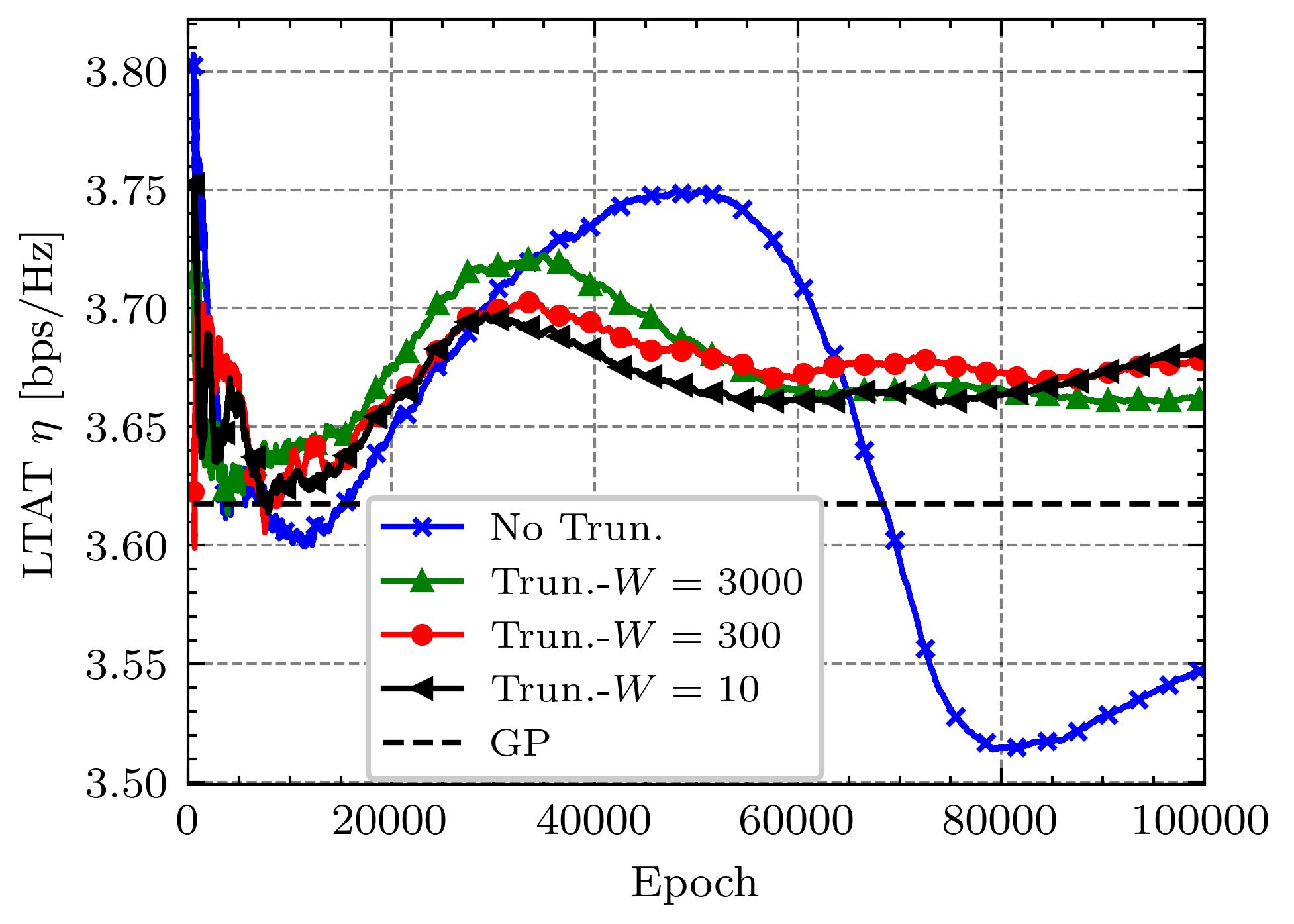}
\caption{The convergence of the LTAT.}
\label{LTAT_conv}
\end{figure}

Fig. \ref{LTA_BLER} plots the convergence curve of BLER versus the number of training epoches. Clearly, Fig. \ref{LTA_BLER} shows that the truncation-based updating approach surpasses the one without truncation in terms of the convergence and the stability of the average BLER curve. Moreover, it is found that the convergence value of the average BLER is lower than the bound $\mathbb{\bar P} = 0.01$.

\begin{figure}[htbp]
\centering
\includegraphics[width=3in]{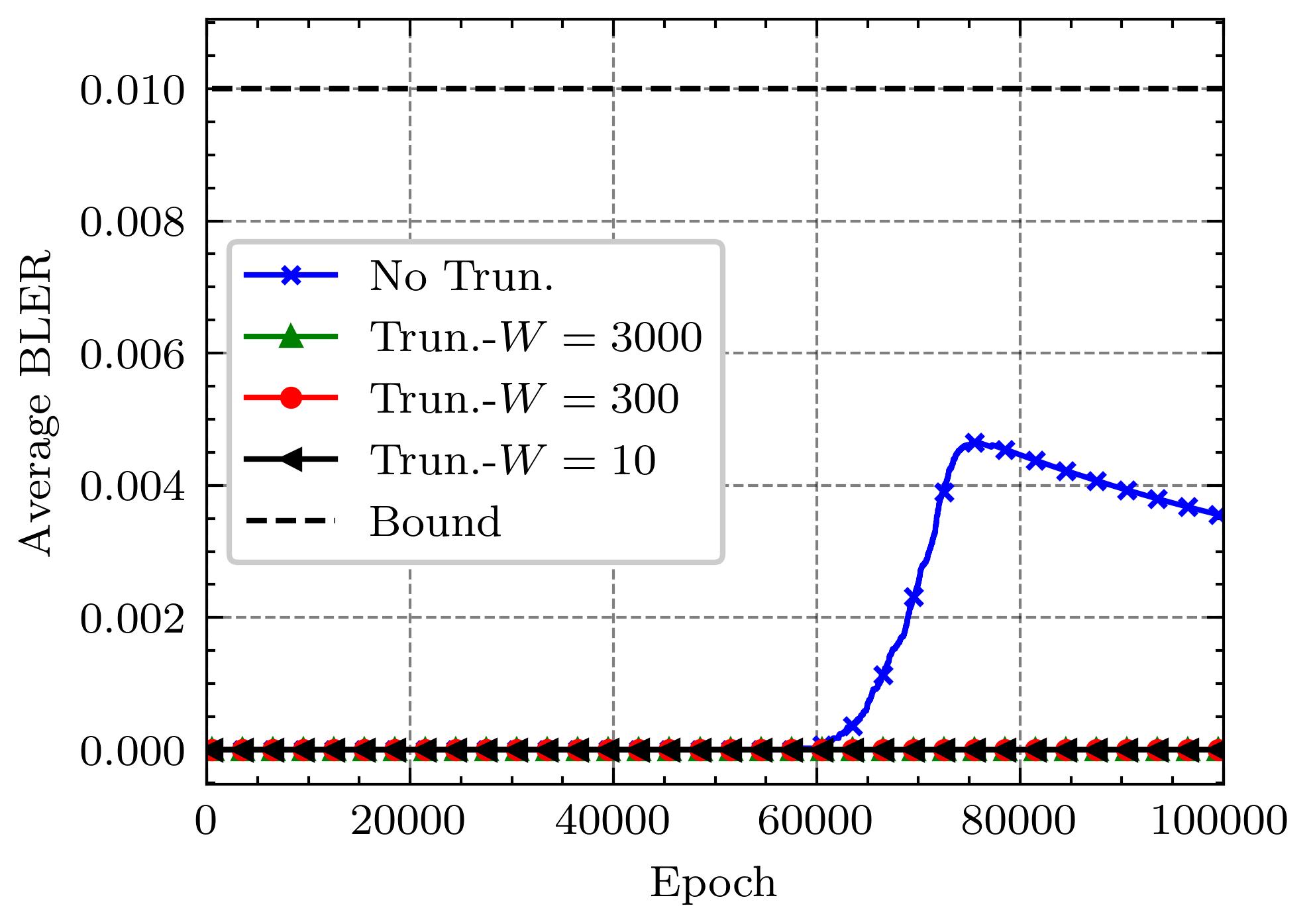}
\caption{The convergence of the average BLER.}
\label{LTA_BLER}
\end{figure}

In addition, the convergence of the average power is plotted against the number of training epoches in Fig. \ref{LTA_Power}. Clearly, it is found that the truncation-based updating approach for $W=300$ shows a more stable convergence than the one without truncation. Besides, it is observed that all the methods slowly converge to the maximum allowable average power. This is not beyond our intuition that the objective function $\eta$ is an increasing function of transmit power in each HARQ round. Accordingly, all the power would be used up to maximize the LTAT $\eta$. 
\begin{figure}[htbp]
\centering
\includegraphics[width=3in]{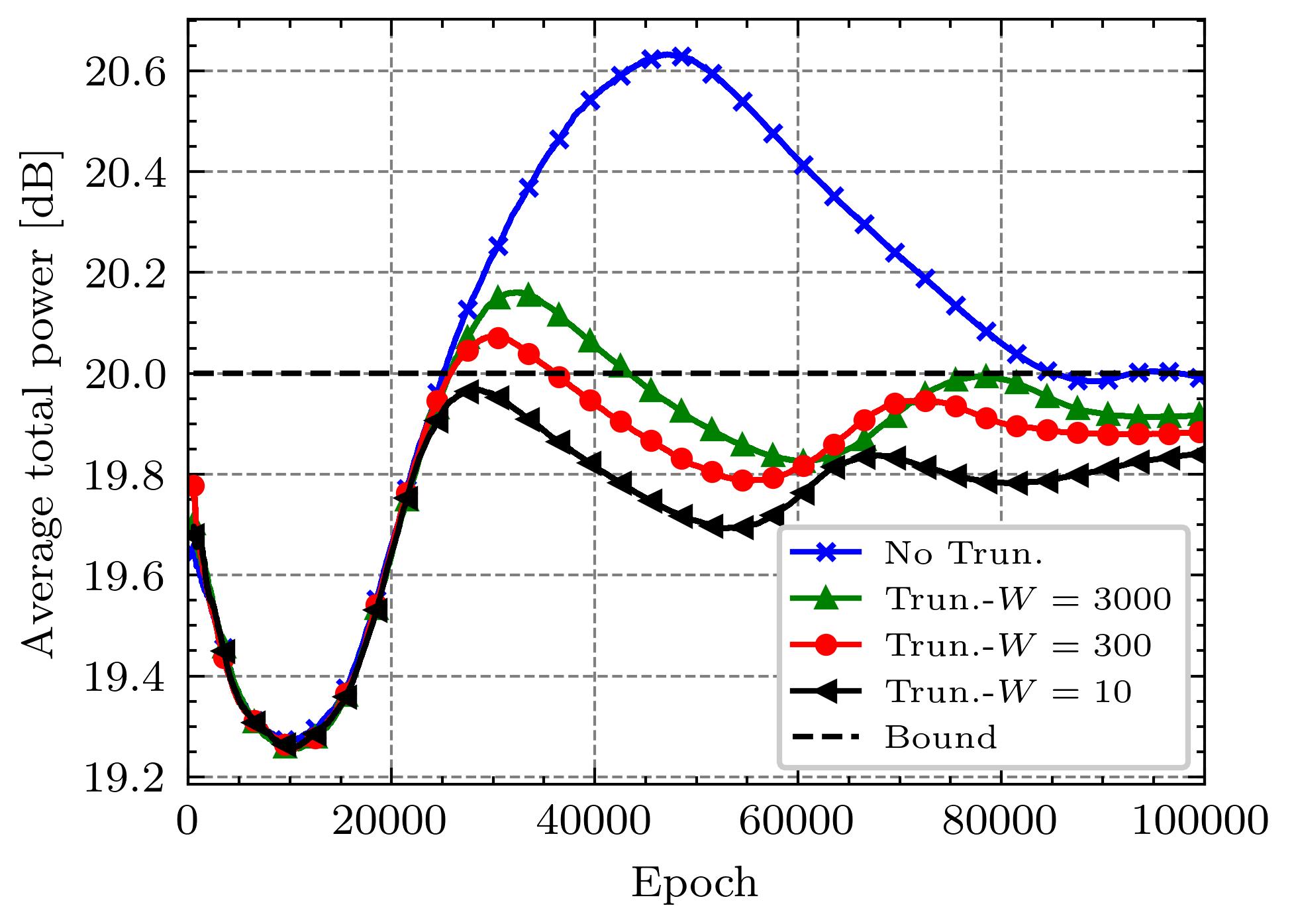}
\caption{The convergence of the average transmission power.} 
\label{LTA_Power}
\end{figure}


\section{Conclusion}\label{CONCULSION}
This paper has examined the average BLER and LTAT maximization of HARQ-IR-aided short packet communications. More specifically, the FBL information theory and the correlated decoding events among different HARQ rounds make it virtually impossible to derive a compact expression for the average BLER. Fortunately, by noticing the recursive form of the average BLER, the trapezoidal approximation method and Gauss–Laguerre quadrature method have been developed to numerically evaluate the average BLER. Besides, the dynamic programming has been applied to implement Gauss–Laguerre quadrature to avoid redundant calculations, consequently yields a reduction of the computational complexity by $(M+1)!$ times. In the meantime, to offer a simple and approximate expression for the average BLER, the asymptotic analysis has been conducted in the high SNR regime. Furthermore, this paper has investigated the power allocation scheme of HARQ-IR-aided short packet communications, which aims to maximize the  LTAT while guaranteeing the power and the BLER constraints. The asymptotic results have enabled solving the optimization problem through the GP. Whereas, the GP-based solution underestimates the system performance at low SNR, which is because of large approximation error of the asymptotic BLER in the circumstance. This has inspired us to utilize DRL to learn power allocation policy from the environment. Towards this end, the original problem has been transformed into a constrained MDP problem, which can be solved by combining DDPG framework and subgradient method. Additionally, the truncation-based updating strategy has been proposed to ensure the stable convergence of the training procedure. The numerical results have finally demonstrated that the DRL-based method is superior to the GP-based one, albeit at the price of high computational burden in offline training stage.

\begin{appendices}

\section{Proof of Corollary \ref{cor1}}\label{appendices:b}
As $N$ tends to infinity, one has
\begin{align}
        \tbinom{M+N}{N} =& \frac{(M+N)!}{N!M!}\notag\\
        \approx& \frac{\sqrt{2\pi (M+N)}\left(\frac{M+N}{e}\right)^{M+N}}{\sqrt{2\pi N}\left(\frac{N}{e}\right)^{N}M!} \notag\\
         \overset{\underset{\mathrm{(a)}}{}}{=}& \sqrt{1+\frac{M}{N}}e^{-M}\left(1 + \frac{M}{N}\right)^N \frac{(M+N)^M}{M!}\notag\\
        \approx & \frac{(M+N)^M}{M!}  \approx \frac{N^M}{M!},
\end{align}
where $n!\approx \sqrt{2\pi n} \left({n}/{e}\right)^n$ by using Stirling's formula and Step (a) holds by using $\lim_{x\to \infty}(1+1/x)^x = e$. 
\section{Proof of Lemma \ref{th1}}\label{appendices:a}
By using the following linearization approximation of $\mathcal Q_m(t)$ \cite[eq.(14)]{ref8}
\begin{equation}
    \begin{aligned}
        \mathcal{Q}_m(t)\simeq&
        \begin{cases}
        1 & t \leq R - {V}_m\\
        \frac{1}{2} - \frac{1}{2{V}_m}(t - R) & R-{V}_m \leq t \leq R+{V}_m\\
        0 & t \geq R + {V}_m
        \end{cases},
    \end{aligned}\label{linear}
\end{equation}
where ${V}_m = \sqrt{\frac{m\pi}{2L}}\log_2 e$. Then by substituting \eqref{linear} into \eqref{high SNR}, $\tilde f_m(x)$ can be obtained as \eqref{Linearization derivation 1}, as shown at the top of the next page,
\begin{figure*}
\begin{align}\label{Linearization derivation 1}
        &\tilde f_{m}(x) 
        \approx\int_{x}^{R - {V}_{m+1}}\tilde f_{m+1}(t)dt +\int_{R-{V}_{m+1}}^{R+{V}_{m+1}}\Big(\frac{1}{2} - \frac{t - R}{2{V}_{m+1}}\Big)\tilde f_{m+1}(t)dt=\notag\\
        & -\tilde F_{m+1}(x) + \frac{1}{2}\bigg(1+\frac{R}{{V}_{m+1}}\bigg)\tilde F_{m+1}(R+{V}_{m+1})+ \frac{1}{2}\bigg(1-\frac{R}{{V}_{m+1}}\bigg)\tilde F_{m+1}(R-{V}_{m+1})- \frac{1}{2{V}_{m+1}}\int\limits_{R-{V}_{m+1}}^{R+{V}_{m+1}}t\tilde f_{m+1}(t)dt,
\end{align}
\hrulefill
\end{figure*}
where $\tilde F_{m+1}(x)=\int_{0}^{x} \tilde f_{m+1}(t)dt$. Besides, by applying integration by parts, the last term in (\ref{Linearization derivation 1}) can be obtained as
\begin{align}\label{eqn:last_ter}
       \int_{R-{V}_{m+1}}^{R+{V}_{m+1}}t\tilde f_{m+1}(t)dt= &(R+{V}_{m+1})\tilde F_{m+1}(R+{V}_{m+1})\notag\\
        & -(R-{V}_{m+1})\tilde F_{m+1}(R-{V}_{m+1})\notag\\
        & - \int_{R-{V}_{m+1}}^{R+{V}_{m+1}}\tilde F_{m+1}(t)dt.
\end{align}
By putting \eqref{eqn:last_ter} into \eqref{Linearization derivation 1} together with some rearrangements, $\tilde f_{m}(x)$ can be finally derived as \eqref{high SNR1}.

\end{appendices}

\bibliographystyle{IEEEtran}
\bibliography{reference}

\begin{thebibliography}{10}
\providecommand{\url}[1]{#1}
\csname url@samestyle\endcsname
\providecommand{\newblock}{\relax}
\providecommand{\bibinfo}[2]{#2}
\providecommand{\BIBentrySTDinterwordspacing}{\spaceskip=0pt\relax}
\providecommand{\BIBentryALTinterwordstretchfactor}{4}
\providecommand{\BIBentryALTinterwordspacing}{\spaceskip=\fontdimen2\font plus
\BIBentryALTinterwordstretchfactor\fontdimen3\font minus
  \fontdimen4\font\relax}
\providecommand{\BIBforeignlanguage}[2]{{%
\expandafter\ifx\csname l@#1\endcsname\relax
\typeout{** WARNING: IEEEtran.bst: No hyphenation pattern has been}%
\typeout{** loaded for the language `#1'. Using the pattern for}%
\typeout{** the default language instead.}%
\else
\language=\csname l@#1\endcsname
\fi
#2}}
\providecommand{\BIBdecl}{\relax}
\BIBdecl

\bibitem{ref30}
H.~Tataria, M.~Shafi, A.~F. Molisch, M.~Dohler, H.~Sjöland, and F.~Tufvesson,
  ``{6G} wireless systems: Vision, requirements, challenges, insights, and
  opportunities,'' \emph{Proc. IEEE}, vol. 109, no.~7, pp. 1166--1199, Mar.
  2021.

\bibitem{ref1}
C.-X. Wang, X.~You, X.~Gao, X.~Zhu, Z.~Li \emph{et~al.}, ``On the road to {6G}:
  Visions, requirements, key technologies, and testbeds,'' \emph{IEEE Commun.
  Surv. Tutor.}, vol.~25, no.~2, pp. 905--974, Feb. 2023.

\bibitem{ref2}
H.~Ding and K.~G. Shin, ``Context-aware beam tracking for {5G} mmwave {V2I}
  communications,'' \emph{IEEE Trans. Mob. Comput.}, vol.~22, no.~6, pp.
  3257--3269, Dec. 2023.

\bibitem{ref29}
Y.~Polyanskiy, H.~V. Poor, and S.~Verdu, ``Channel coding rate in the finite
  blocklength regime,'' \emph{IEEE Trans. Inf. Theory}, vol.~56, no.~5, pp.
  2307--2359, Apr. 2010.

\bibitem{ref47}
J.~Zheng, Q.~Zhang, and J.~Qin, ``Average achievable rate and average {BLER}
  analyses for {MIMO} short-packet communication systems,'' \emph{IEEE Trans.
  Veh. Technol.}, vol.~70, no.~11, pp. 12\,238--12\,242, Oct. 2021.

\bibitem{ref49}
Z.~Shi, H.~Zhang, H.~Wang, Y.~Fu, G.~Yang, X.~Ye, and S.~Ma, ``Block error rate
  analysis of short-packet mobile-to-mobile communications over correlated
  cascaded fading channels,'' \emph{IEEE Trans. Veh. Technol.}, vol.~71, no.~4,
  pp. 4087--4101, Feb. 2022.

\bibitem{ref50}
T.-H. Vu, T.-T. Nguyen, Q.-V. Pham, D.~B. da~Costa, and S.~Kim, ``A novel
  partial decode-and-amplify {NOMA}-inspired relaying protocol for uplink
  short-packet communications,'' \emph{IEEE Wireless Commun. Lett.}, vol.~12,
  no.~7, pp. 1244--1248, Apr. 2023.

\bibitem{ref51}
N.~P. Le and K.~N. Le, ``Uplink {NOMA} short-packet communications with
  residual hardware impairments and channel estimation errors,'' \emph{IEEE
  Trans. Veh. Technol.}, vol.~71, no.~4, pp. 4057--4072, Feb. 2022.

\bibitem{ref52}
T.-H. Vu, T.-V. Nguyen, D.~B.~d. Costa, and S.~Kim, ``Intelligent reflecting
  surface-aided short-packet non-orthogonal multiple access systems,''
  \emph{IEEE Trans. Veh. Technol.}, vol.~71, no.~4, pp. 4500--4505, Jan. 2022.

\bibitem{ref53}
J.~Yang, L.~Zhao, J.~Ding, and Z.~Ding, ``Performance analysis of short-packet
  {NOMA} systems assisted by {IRS} with failed elements,'' \emph{IEEE Trans.
  Veh. Technol.}, pp. 1--6, Nov. 2023, doi={10.1109/TVT.2023.3330582}.

\bibitem{ref19}
C.~Yue, V.~Miloslavskaya, M.~Shirvanimoghaddam, B.~Vucetic, and Y.~Li,
  ``Efficient decoders for short block length codes in {6G} {URLLC},''
  \emph{IEEE Commun. Mag.}, vol.~61, no.~4, pp. 84--90, Apr. 2023.

\bibitem{ref5}
Y.~Yang, Y.~Song, and F.~Cao, ``{HARQ} assisted short-packet communications for
  cooperative networks over {Nakagami-m} fading channels,'' \emph{IEEE Access},
  vol.~8, pp. 151\,171--151\,179, Aug. 2020.

\bibitem{ref10}
F.~Nadeem, Y.~Li, B.~Vucetic, and M.~Shirvanimoghaddam, ``Real-time wireless
  control with non-orthogonal {HARQ},'' in \emph{2022 IEEE GLOBECOM Workshops
  (GC Wkshps'22)}, Rio de Janeiro, Brazil, Dec. 2022, pp. 395--400.

\bibitem{ref6}
D.~Marasinghe, N.~Rajatheva, and M.~Latva-Aho, ``Block error performance of
  {NOMA} with {HARQ-CC} in finite blocklength,'' in \emph{Proc. IEEE Int. Conf.
  Commun. Workshops (ICC Workshops'20)}, Dublin, Ireland, Jnu. 2020.

\bibitem{ref9}
F.~Ghanami, G.~A. Hodtani, B.~Vucetic, and M.~Shirvanimoghaddam, ``Performance
  analysis and optimization of {NOMA} with {HARQ} for short packet
  communications in massive {IoT},'' \emph{IEEE Internet Things J.}, vol.~8,
  no.~6, pp. 4736--4748, Oct. 2021.

\bibitem{ref8}
B.~Makki, T.~Svensson, and M.~Zorzi, ``Finite block-length analysis of the
  incremental redundancy {HARQ},'' \emph{IEEE Wireless Commun. Lett.}, vol.~3,
  no.~5, pp. 529--532, Aug. 2014.

\bibitem{ref36}
A.~Avranas, M.~Kountouris, and P.~Ciblat, ``Energy-latency tradeoff in
  ultra-reliable low-latency communication with retransmissions,'' \emph{IEEE
  J. Sel. Areas Commun.}, vol.~36, no.~11, pp. 2475--2485, Nov. 2018.

\bibitem{ref35}
{Avranas, Apostolos and Kountouris, Marios and Ciblat, Philippe}, ``Throughput
  maximization and {IR-HARQ} optimization for {URLLC} traffic in {5G}
  systems,'' in \emph{Proc. IEEE Int. Conf. Commun. (ICC'19)}, Shanghai, China,
  May 2019, pp. 1--6.

\bibitem{ref34}
N.~Li and D.~Qiao, ``Resource optimization for noncoherent short-packet
  communications with {IR-HARQ},'' in \emph{Proc. 13th Int. Conf. Wirel.
  Commun. Signal Process. (WCSP'21)}, Virtual, Online, China, Dec. 2021, pp.
  1--5.

\bibitem{ref27}
J.-H. Park and D.-J. Park, ``A new power allocation method for parallel {AWGN}
  channels in the finite block length regime,'' \emph{IEEE Commun. Lett.},
  vol.~16, no.~9, pp. 1392--1395, Jul. 2012.

\bibitem{ref32}
B.~Makki, T.~Svensson, T.~Eriksson, and M.-S. Alouini, ``On the required number
  of antennas in a point-to-point large-but-finite {MIMO} system:
  Outage-limited scenario,'' \emph{IEEE Trans. Commun.}, vol.~64, no.~5, pp.
  1968--1983, 2016.

\bibitem{ref37}
K.~Atkinson, \emph{An introduction to numerical analysis}, 2nd~ed.\hskip 1em
  plus 0.5em minus 0.4em\relax New York: John Wiley and Sons, 1991.

\bibitem{ref38}
H.~E. Salzer and R.~Zucker, ``Table of the zeros and weight factors of the
  first fifteen laguerre polynomials,'' \emph{Bull. Am. Meteorol. Soc.},
  vol.~55, pp. 1004--1012, Oct. 1949.

\bibitem{ref39}
I.~Gradshteyn, I.~Ryzhik, V.~Geronimus~Yu, M.~Y. Tseytlin, and J.~Alan,
  \emph{Table of Integrals, Series, and Products}, 8th~ed.\hskip 1em plus 0.5em
  minus 0.4em\relax San Diego, CA, USA: Academic, 2014.

\bibitem{ref42}
G.~Caire and D.~Tuninetti, ``The throughput of hybrid-{ARQ} protocols for the
  {Gaussian} collision channel,'' \emph{IEEE Trans. Inf. Theory}, vol.~47,
  no.~5, pp. 1971--1988, Jul. 2001.

\bibitem{ref41}
D.~P. Bertsekas, ``Nonlinear programming,'' \emph{J. Oper. Res. Soc.}, vol.~48,
  no.~3, pp. 334--334, Jan. 1997.

\bibitem{ref43}
D.~Wu, J.~Feng, Z.~Shi, H.~Lei, G.~Yang, and S.~Ma, ``Deep reinforcement
  learning empowered rate selection of {XP-HARQ},'' \emph{IEEE Commun. Lett.},
  vol.~27, no.~9, pp. 2363--2367, Jul. 2023.

\bibitem{ref14}
T.~P. Lillicrap, J.~J. Hunt, A.~Pritzel, N.~Heess, T.~Erez, Y.~Tassa,
  D.~Silver, and D.~Wierstra, ``\BIBforeignlanguage{English}{Continuous control
  with deep reinforcement learning},'' in
  \emph{\BIBforeignlanguage{English}{Proc. 4th Int. Conf. Learn. Represent.
  (ICLR'16)}}, San Juan, Puerto rico, May 2016, pp. 1--14.

\bibitem{ref44}
T.~{Schaul}, J.~{Quan}, I.~{Antonoglou}, and D.~{Silver}, ``Prioritized
  experience replay,'' in \emph{Proc. Int. Conf. Learn. Represent. (ICLR'16)},
  San Juan, Puerto Rico, May 2016.

\bibitem{ref45}
D.~Silver, G.~Lever, N.~Heess, T.~Degris, D.~Wierstra, and M.~Riedmiller,
  ``Deterministic policy gradient algorithms,'' in \emph{Proc. 31st Int. Conf.
  Mach. Learn. (ICML'14)}, Beijing, China, Jun. 2014, pp. 387--395.

\end{thebibliography}

\end{document}